\documentclass{article}
\usepackage{url}
\usepackage{amsthm,amsmath,amsfonts}
\newtheorem{thm}{Theorem}

\newtheorem{lem}[thm]{Lemma}

\newtheorem{prop}[thm]{Proposition}

\theoremstyle{definition}
\newtheorem{defn}[thm]{Definition}
\newtheorem{exmp}[thm]{Example}

\makeatletter
\newdimen\@bls                              
\@bls=\baselineskip                         
\newenvironment{pf}%
  {\par\addvspace{\@bls \@plus 0.5\@bls \@minus 0.1\@bls}\noindent
   {\bfseries\Elproofname}\enspace\ignorespaces}%
  {\par\addvspace{\@bls \@plus 0.5\@bls \@minus 0.1\@bls}}
\def\Elproofname{PROOF.}
\makeatother

\usepackage{natbib}
\usepackage{simredkron} 

\begin{document}
\title{Simultaneous Modular Reduction and Kronecker Substitution for
  Small Finite Fields\thanks{This research was partly supported by the
    French National Research Agency (ANR Safescale, ANR Gecko)}
}

\author{Jean-Guillaume Dumas\footnote{Laboratoire J. Kuntzmann,
    Universit\'e de Grenoble, umr CNRS 5224. BP 53X, 51, rue des
    Math\'ematiques, F38041 Grenoble, France. \{Jean-Guillaume.Dumas,Laurent.Fousse\}@imag.fr.}
\and Laurent Fousse\footnotemark[2]
\and Bruno Salvy\footnote{Algorithms Project,
INRIA Rocquencourt,
78153 Le Chesnay.
France. Bruno.Salvy@inria.fr..}}
\maketitle
\begin{abstract} We present algorithms to perform modular polynomial
  multiplication or modular dot product efficiently in a single
  machine word. We pack
  polynomials into integers and perform several modular operations with
  machine integer or floating point arithmetic. 
  The modular 
  polynomials are converted into integers using Kronecker substitution (evaluation at a sufficiently large integer).
  With some control on the sizes and degrees, arithmetic operations on the polynomials can be performed
  directly with machine integers or floating point numbers and the number of 
  conversions can be reduced.
  We also present efficient ways to recover the modular values of the coefficients.
  This leads to practical gains of quite large constant factors for polynomial
  multiplication, prime field linear algebra and small extension field
  arithmetic.
\end{abstract}


\section{Introduction}
While theoretically well understood, the basic routines for linear algebra or polynomial arithmetic over finite fields are difficult to implement efficiently. Important factors of speed can be gained by exploiting machine integer or floating-point arithmetic and taking cache effects into account. This has been demonstrated for instance by~\citeauthor{jgd:2002:fflas} (\citeyear{jgd:2002:fflas,jgd:2004:ffpack}) who wrapped cache-aware routines for efficient small finite field linear algebra in their FFLAS/FFPACK project.

The elements of small enough prime fields can be represented as integers fitting in a machine word or half-word, or in exact floating point numbers. For extension fields, the elements can be represented as polynomials over a prime field. A further compression is proposed by~\citeauthor{jgd:2002:fflas} (\citeyear{jgd:2002:fflas,jgd:2008:issac}), using Kronecker substitution~\citep[\S 8.4]{VonzurGathen:1999:MCA}: the polynomials are represented by their value at an integer~$q$ larger than the characteristic of the field. We call \dqt, for Discrete Q-adic Transform,  this simple map:
\begin{equation}\label{eq:dqt}
\begin{split}	
	\operatorname{DQT}:\Z/{p\Z}[X]&\rightarrow\Z\\
	\sum_{i=0}^k{\alpha_iX^i}&\mapsto\sum_{i=0}^k{\alpha_iq^i}.
\end{split}
\end{equation}
%
With some care, in particular on
the size of $q$, it is possible to map the operations in the extension
field into the floating point arithmetic realization of this $q$-adic
representation. 

In computational number theory, 
matrices over $\mathbb{F}_2$ are often compressed by fitting several entries into one 
via the binary representation of machine integers~\citep{Coppersmith:1993:SLE,Kalto-Lobo}.
The need for efficient matrix computations over {very small} finite fields also arises in other areas and in particular graph theory (adjacency matrices), see e.g., 
\citep{May:2007:graphs} or 
\citep{Weng:2007:paley}. In these cases, one is thus led to work with polynomials as in~\eqref{eq:dqt}, using a small~$q$.

Recovering the results is obtained by an inverse \dqt. This consists of a radix conversion
\cite[Algorithm 9.14]{VonzurGathen:1999:MCA}, followed by a reduction  of the coefficients modulo~$p$.
On modern processors, machine division or
remaindering, that are used in radix conversion, are quite slow when compared to other arithmetic
operations\footnote{On a Xeon 3.6GHz using doubles, addition, multiplication and axpy take roughly the same time, while division is 10 times slower and fmod (floating point remainder) again 2.5 times as long as division}. 

In this article, we use two techniques to avoid some of these remainderings:
\emph{(i)} delayed reduction; \emph{(ii)} simultaneous reduction.
Delayed reduction means that after having been transformed into their \dqt form, polynomials may undergo several arithmetic operations before being converted back into coefficient form. This is possible as long as the coefficients remain smaller than~$q$, which is prevented by \emph{a priori} bounds. Simultaneous reduction is the operation of recovering~$\alpha_i\bmod p$ from~$\alpha_i$, for $i=1,\dots,k$, in the context of the inverse \dqt. 
We propose a new
algorithm called REDQ performing these $k$ modular reductions by a single division,
$\lceil \frac{k+1}{2} \rceil$ additions and multiplications and some
table look-up. We also discuss the possibility of replacing the remaining division by floating point operations, taking into account the rounding modes.

We recall in Section~\ref{sec:galois} the Kronecker
substitution and delayed reduction algorithms. 
Then we present our new simultaneous reduction algorithm and give its
complexity in Section~\ref{sec:redq} and we discuss how to replace
the remaining machine division by floating point operations with
different rounding modes in Section~\ref{sec:floor}.
Our new reduction algorithm has two parts. 
The first one is a \emph{compression},
performed by arithmetic operations, which reduces the size of the
polynomial entries. 
The second one is performed only when required and is a {\em
  correction}, 
which gives to the residues their
correct value modulo $p$  when the compression has shifted the result. 
We also present a time-memory trade-off enabling
some table look-ups computing this correction. 
Then we apply the \dqt~to different contexts:
modular polynomial multiplication in Section~\ref{sec:del};
linear algebra over small extension fields in Section~\ref{sec:fflas}; 
compressed linear algebra over small prime fields in Section~\ref{sec:cmm}.
This gives some constraints on the possible choices
of $q$ and $k$.
In the three applications anyway, we show that these compression techniques
represent a speed-up factor usually of the order of the number
$k$ of residues stored in the compressed format. 

Preliminary versions of this work have been presented by
\cite{jgd:2008:issac} and \cite{jgd:2008:mica}. Here, we give an
improved version of the simultaneous reduction where the number of
operations has been divided by two. We also give a complete study of
the behavior of the division of integers by
floating point routines, depending on the rounding modes. 
Finally, we present more experimental
results and faster implementations of the applications, namely a
Karatsuba version of the polynomial multiplication and a right
compressed matrix multiplication.




\section{Q-adic Representation of Polynomials}\label{sec:galois}
\subsection{Kronecker Substitution}              
The principle of Kronecker substitution is very simple.
It consists in evaluating polynomials at a given integer as in Eq.~\eqref{eq:dqt}.
For instance, for $k=2$, the substitution is performed by the following compression:
\begin{verbatim}
double& init3( double& r, const double u, const double v, const double w) {
        // _dQ is a floating point storage of Q
        r=u; r*=_dQ; r+=v; r*=_dQ; return r+=w;
}
\end{verbatim}
   
The integer $q$ can be chosen to be a power of 2 in a binary
architecture. Then the
Horner like evaluation 
of a polynomial at $q$
is just a left shift. One can then compute
this shift with exponent manipulations in floating point arithmetic
and use native shift operators (e.g., the $<\!\!<$ operator in C)
as soon as values are within
the $32$ (or $64$ when available) bit range. 

\medskip
The motivation for this substitution is its use in multiplication.
\begin{exmp}
To multiply $a=X+1$ by $b=X+2$ in $\pF{3}[X]$ one can use
the substitution $X=q:=100$; compute $101 \times 102 = 10302$; use radix
conversion to write $10302=q^2+3q+2$; reduce the coefficients modulo $3$ to get $a
\times b = X^2+2$.
\end{exmp}

More generally, if $p$ is prime,  $a = \sum_{i=0}^{k-1} a_i X^i$
and $b = \sum_{i=0}^{k-1} b_i X^i$
are two polynomials in $\pF{p}[X]$, then one can perform
the polynomial multiplication $ a b $ via Kronecker substitution.
The product of 
$\tilde{a} = \sum_{i=0}^{k-1} a_i q^i$
and $\tilde{b} = \sum_{i=0}^{k-1} b_i q^i$ is given by
\begin{equation}\label{eq:mult}
\widetilde{a b} = \sum_{j=0}^{2k-2} \left( \sum_{i=0}^{j} a_i
b_{j-i} \right) q^j.
\end{equation}
Now if $q$ is large enough, the coefficient of $q^j$ does not exceed $q-1$. If moreover~$k$ is not too large, the product fits in a machine number (floating point number or integer).
Thus in that case, it is possible to evaluate $\tilde{a}$ and $\tilde{b}$ as machine
numbers, compute the product of these evaluations, and convert back to
polynomials by radix conversion. There just remains to perform
reductions of the coefficients modulo $p$.

We show in Section~\ref{sec:fflas} that one can also tabulate the
evaluations at $q$, and that one can access directly the required part
of the machine words (using e.g., bit fields and unions in C)
instead of performing a radix conversion. 

\subsection{Delayed Reduction}
With current processors, machine division (as well as modular reduction) 
is still much slower in general than
machine addition and machine multiplication. 
It is possible to replace the machine division by some other implementation, such as floating
point multiplication by the inverse \citep{Shoup:2007:NTL} or Montgomery reduction \citep{Montgomery:1985:MMT}; see
e.g., \citep{jgd:2004:dotprod} and references therein for more details.

It is also important to reduce the number of machine remainderings when performing modular
computations. This is achieved by using Kronecker substitution seldom and perform as many arithmetic operations as possible on the transformed values.

\subsection{Discrete Q-adic Transform}
The idea of the Q-adic transform is to combine Kronecker
substitution with delayed reduction. 

We call \dqt the evaluation of polynomials modulo $p$ at a
sufficiently large $q$. Since this is a ring morphism, products as well as additions can be performed on the transformed values.
We call \dqt inverse the radix conversion of a
$q$-adic expansion
followed by a modular reduction. For a given procedure $\star(a_1,\dots,a_m)$ performing only ring operations, the idea is to compute the \dqt of the $a_i$'s, perform $\star$ on these \dqt's and compute the inverse \dqt at the end. For appropriate choices of the parameters, this procedure recovers the exact value. As an example, we recall the dot product of~\cite{jgd:2002:fflas} in Algorithm~\ref{alg:dqt}.
\newcounter{horner}
\newcounter{radix}
\newcounter{modin}
\begin{algorithm}[ht]
\caption{Polynomial dot product by \dqt}\label{alg:dqt}
\begin{algorithmic}[1]
\REQUIRE Two vectors of polynomials $v_1$ and $v_2$ in $\pF{p}[X]^n$ of
degree less than $k$;
\REQUIRE a sufficiently large integer $q$.
\ENSURE $R \in \pF{p}[X] $, with $R = v_1^T\cdot v_2$.
\vspace{5pt}\newline{\underline{Polynomial to $q$-adic conversion}}\vspace{2pt}
\STATE Set $\widetilde{v_1}$ and $\widetilde{v_2}$ to the floating point vectors
of the evaluations at $q$ of the elements of $v_1$ and $v_2$.
\setcounter{horner}{\value{ALC@line}}
\COMMENT{Using e.g., Horner's formula}
\vspace{5pt}\newline{\underline{One computation}}\vspace{2pt}
\STATE Compute $\tilde{r} = \widetilde{v_1}^T\cdot\widetilde{v_2}$
\vspace{5pt}\newline{\underline{Building the solution}}\vspace{2pt}
\STATE $\tilde{r} = \sum_{i=0}^{2k-2} \widetilde{\mu_i} q^i$.
\setcounter{radix}{\value{ALC@line}}
\COMMENT{Using radix conversion} 
\STATE For each $i$, set $\mu_i = \widetilde{\mu_i} \bmod p$
\STATE Return $R = \sum_{i=0}^{2k-2} \mu_i X^i $
\setcounter{modin}{\value{ALC@line}}
\end{algorithmic}
\end{algorithm}       

The following bounds generalize those of
\citep[\S 8.4]{VonzurGathen:1999:MCA}:
\begin{thm}{\citep{jgd:2002:fflas}}
Let $\beta$ be the number of  mantissa bits
available within the machine numbers. 
If
\begin{equation}\label{eq:bounds}
q > n k (p-1)^2\quad\text{and}\quad (2k-1)\log_2(q) \leq \beta,
\end{equation} then
Algorithm~\ref{alg:dqt} is correct.
\end{thm}

\citeauthor{jgd:2002:fflas} (\citeyear{jgd:2002:fflas}, Figures 5 \& 6)
show that this wrapping is already a pretty good way to obtain high speed
linear 
algebra over small extension fields. 
They reach  high peak performance, quite close to those obtained with
prime fields, namely 420 Millions of finite field operations per second
(Mop/s) on a Pentium III, 735 MHz, and more than 500
Mop/s on a 64-bit DEC alpha 500 MHz. This is roughly 20\% below
the pure floating point performance and 15\% below the prime
field implementation. We  show in Section~\ref{sec:fflas} that the
new algorithms presented in this article enable to reduce this overhead
to less than 4\%.




\section{REDQ: Modular Reduction in the \dqt Domain}\label{sec:redq}
The first improvement we propose to the \dqt is to replace the costly
modular reduction of the polynomial coefficients (e.g., in Step~4 of Algorithm~\ref{alg:dqt}) by a \emph{single} division
by $p$ followed by several
shifts. In the next section, we replace this division by a multiplication by an inverse.

\subsection{Examples}
We first illustrate the basic idea on a simple example.
\begin{exmp}
Let $a=X^2+2X+3$ and $b=4X^2+5X+6$ be unreduced modulo $5$, so that we want to compute $a\times b = 4X^4+3X^3+3X^2+2X+3$.
We are free to choose the integer~$q$ at which the evaluation takes place in such a way that shifting by powers of~$q$ is cheap. For clarity we take here~$q=10000$.
Thus, we start from $\tilde{r}:=\tilde{a}\times\tilde{b}=\mathbf{4}00\mathbf{13}00\mathbf{28}00\mathbf{27}00\mathbf{18}$
for which
we need to reduce five coefficients modulo $5$. 
In the direct approach, the coefficients would be recovered by computing
\[4=0\times5+4,\quad 13=2\times5+3,\quad 28=5\times5+3,\quad27=5\times5+2,\quad18=3\times5+3.\]
The trick
is that we can recover all the quotients at once.
Indeed, we compute $s:=\lfloor \tilde{r}/p\rfloor=\mathbf{0}800\mathbf{2}600\mathbf{5}600\mathbf{5}400\mathbf{3}$ that contains
all the quotients (in boldface). The remainders (the coefficients)
can then be recovered as~$u_i:=\lfloor\tilde{r}/q^i\rfloor-p\lfloor s/q^i\rfloor$, for $i=0,\dots,4.$
\end{exmp}

Some more work is needed in order to make this idea correct in general.
\begin{exmp}\label{ex:full}We now consider 
%
%
%
%
the polynomial $R=1234X^3+5678X^2+9123X+4567$, the
prime $p=23$ and use $q=10^6$. Note that this time, $p$ does not divide~$q$.  
The polynomial we want to recover is $R\bmod p=15X^3+20X^2+15X+13$.
We start from~$\tilde{r}=1234005678009123004567$ and 
the division gives $s :=\lfloor \tilde{r}/p\rfloor=53652420783005348024$.
Proceeding as before, we get
\[u_0=15,\quad u_1=8,\quad u_2=18,\quad u_3=15.\]
These coefficients are small, but they are not the correct ones except for~$u_3$. 
Indeed, if $C=\alpha p+u_i$, then $C q=u_i q\neq0 \bmod p$ so that each coefficient needs to be corrected to take into account the values of the preceding ones.
We thus consider this first computation as a \emph{compression} stage and use a \emph{correction} stage that produces the correct values.
The correction is obtained by taking~$\mu_3=u_3$ and~$\mu_i=u_i-qu_{i+1}\bmod p$ for~$i=0,1,2$ and returning the~$\mu_i$'s.
%
\end{exmp}

\subsection{Algorithm}

\makeatletter
\newcounter{linerop}\newcounter{lineui}
\newcounter{lineqdend}\newcounter{lineqdbeg}\newcounter{linemui}
\begin{algorithm}[ht]
\caption{REDQ}\label{alg:REDQ}
\begin{algorithmic}[1]
\REQUIRE Two integers $p$ and $q$ satisfying Conditions (\ref{eq:bounds});
\REQUIRE $\tilde{r}  = \sum_{i=0}^d \widetilde{\mu_i} q^i \in \Z$.
\ENSURE $\rho \in \Z$, with $\rho = \sum_{i=0}^d \mu_i q^i$
where $\mu_i = \widetilde{\mu_i} \bmod p$.
\UNDERL{REDQ COMPRESSION}

\STATE\setcounter{linerop}{\theALC@line} $s = \left\lfloor \frac{\tilde{r}}{p} \right\rfloor$;
\FOR{$i=0$ to $d$}
\STATE\setcounter{lineui}{\theALC@line} $u_i = \left\lfloor \frac{\tilde{r}}{q^i} \right\rfloor -  p \left\lfloor \frac{s}{q^i} \right\rfloor$;
\ENDFOR
\UNDERL{REDQ CORRECTION} \COMMENT{when $p\nmid q$}
\STATE\setcounter{lineqdbeg}{\theALC@line} $\mu_{d}=u_{d}$
\FOR{$i=0$ to $d-1$}
\STATE\setcounter{linemui}{\theALC@line} $\mu_i = u_i-qu_{i+1} \bmod p$;
\ENDFOR\setcounter{lineqdend}{\theALC@line}
\STATE Return $\rho = \sum_{i=0}^d \mu_i q^i$;
\end{algorithmic}
\end{algorithm}
\makeatother

\begin{thm}\label{thm:REDQ}
Algorithm REDQ is correct.
\end{thm}
The following lemma is probably classical. We state and prove it here for completeness.
\begin{lem}\label{lem:flfl} For $r  \in \N$ and $a$, $b \in \N^*$,
$$\left\lfloor \frac{\left\lfloor
      \frac{r}{b}\right\rfloor}{a}\right\rfloor =
\left\lfloor \frac{r}{ab}\right\rfloor =
\left\lfloor \frac{\left\lfloor \frac{r}{a}\right\rfloor}{b}\right\rfloor.$$
\end{lem}
\begin{pf}  
	Let $k=\left\lfloor {r}/{ab}\right\rfloor$, so that
$kab \leq r <  (k+1)ab$.
Then $kb \leq {r}/{a} < (k+1)b$ and since
$kb$ is an integer it follows that $kb \leq \left\lfloor
  {r}/{a}\right\rfloor < (k+1)b$. Dividing by~$b$ yields $\left\lfloor
  {\left\lfloor {r}/{a}\right\rfloor}/{b}\right\rfloor
=k$. The other side of the identity follows by exchanging~$a$ and $b$.
\end{pf}
\begin{pf}[of Theorem~\ref{thm:REDQ}]
First we prove that $0 \leq u_i < p$. 
Let~$T_i=\lfloor\tilde{r}/q^i\rfloor$. By Lemma~\ref{lem:flfl} $\lfloor s/q^i\rfloor=\lfloor T_i/p\rfloor$, so that $u_i=T_i-p\lfloor T_i/p\rfloor$ which proves the inequalities.
%

Next we compute the value of~$u_i$ by 
the following sequence of identities modulo~$p$:
\[u_i=T_i=\left\lfloor\frac{\tilde{r}}{q^i}\right\rfloor=\sum_{j=i}^{d}
\tilde{\mu}_j q^{j-i}=\sum_{j=i}^{d} \mu_j q^{j-i} \bmod p.\]
When~$q=0\bmod p$, we thus have $u_i=\mu_i\bmod p$ and the proof is complete.
Otherwise, in view of the previous identity, the correction step on Line~\thelinemui\ produces
\[u_i-q u_{i+1}=\sum_{j=i}^{d} \mu_j q^{j-i}-q\sum_{j=i+1}^{d} \mu_j q^{j-i-1}=\mu_i\bmod p.\]
%
\end{pf}



\begin{defn} We call REDQ$_k$ a simultaneous 
reduction of $k$ residues performed by Algorithm~\ref{alg:REDQ} 
(in other words $k=d+1$ if $d$ is the degree of the Kronecker
substitution).
\end{defn}

\subsection{Binary Case}
When $q$ is a power of $2$ and elements are represented using an
integral type, division by $q^i$ and flooring are a single
operation, a right shift, or direct bit fields extractions when available.
Moreover,
the
remainders can be computed independently and thus the loop of
REDQ\_COMPRESSION can be performed by only half of the required $k$
axpy (combined addition and multiplication, or fused-mac) as shown
below:
\begin{figure}[htbp]\center
\includegraphics[width=\textwidth*7/8]{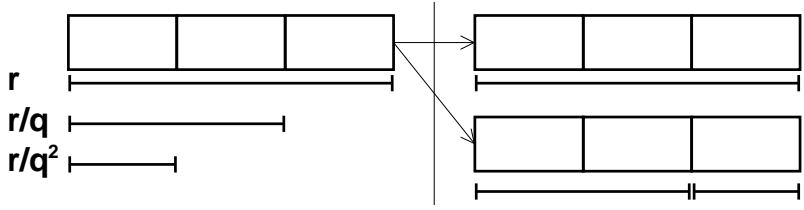}
\caption{REDQ$_3$\_COMPRESSION with 2 axpy}
\end{figure}

\begin{prop}\label{prop:redqcost}
Let $q$ be a power of two. Then, 
a REDQ$_k$\_COMPRESSION requires $\lceil \frac{k+1}{2} \rceil$
  machine word multiplications and additions.
\end{prop}
\begin{pf}
We use the notations of Algorithm~\ref{alg:REDQ}. Let 
$b_q=\log_2(q)$ be the number of bits of $q=2^{b_q}$ and let $\beta$
be the number of bits in the mantissa of a machine word.
If the REDQ$_k$ is correct, we have $q^k\leq 2^\beta$, or $b_q k \leq \beta$. 
The first value $u_0 = \tilde{r} - s \times p$ requires the whole mantissa. 
Now both $\lfloor \frac{\tilde{r}}{q^i} \rfloor$ and $\lfloor
\frac{s}{q^i} \rfloor$ can be stored on at most $k \times b_q -i
\times b_q$
bits. Moreover, by Theorem~\ref{thm:REDQ},
$0 \leq \lfloor \frac{\tilde{r}}{q^i} \rfloor- p\times \lfloor
\frac{s}{q^i} \rfloor <\lfloor \frac{\tilde{r}}{q^i} \rfloor$ 
so that the result does not overflow these $(k-i)b_q$
bits. This proves that the operations can be done independently on the
different parts of the machine words. Now, the total number of bits
required to perform the loop
of REDQ\_COMPRESSION is 
\[ b_q \sum_{i=0}^{k-1}  k-i = b_q \frac{k(k+1)}{2} \leq \beta
\left\lceil \frac{k+1}{2} \right\rceil.
\] 
This, combined with the non-overlapping of the operations, proves the
proposition.
\end{pf}

Here is an instance of a REDQ compression with 3 residues in C++,
using the syntactic sugar of bit field extraction:
\begin{verbatim}
inline void REDQ_COMP(UINT32_three& res,
                      const double r,                  // to be reduced
                      const double p){                 // modulo
    _ULL64_unions rll, tll;           // union of 64, 17/34 or 34/17 bits 
    rll._64 = static_cast<UINT64>( r );
    tll._64 = static_cast<UINT64>( r/p );              // One division 
    res.high = static_cast<UINT32>(rll._64-tll._64*p); // One axpy
    rll._17_34.low = rll._34_17.high;                  // Packing
    tll._17_34.low = tll._34_17.high;                  // Packing
    rll._64 -= tll._64*p;                              // Two axpy in one
    res.mid = static_cast<UINT32>(rll._17_34.high);
    res.low = rll._17_34.low;
}
\end{verbatim}

In general also the algorithm is efficient because one can precompute
$1/p$, $1/q$, 
$1/q^2$ etc. The
computation of each $u_i$ and $\mu_i$ can also be pipelined or
vectorized since they are independent.
As is, the benefit when
compared to direct remaindering by $p$ is that the corrections occur
on smaller integers. Thus the remaindering by $p$ can be faster.
Actually, another major acceleration can be added:
the fact that the $\mu_i$ are much smaller than the initial
$\tilde{\mu_i}$ makes it possible to tabulate the corrections as shown
next.


\subsection{Matrix Version of the Correction}     
In Example~\ref{ex:full}, the final correction can be written
as a matrix vector product:
\[ \mu = \left[ \begin{array}{cccc} 1 & 17 & 0 &
    0\\0&1 &17&0\\0&0&1&17\\0&0&0&1\end{array}\right] u \bmod p.
\]

More generally, the corrections of lines \thelineqdbeg~to
\thelineqdend~of Algorithm~\ref{alg:REDQ} are given 
by a matrix-vector multiplication with an invertible matrix $Q_k$:
{\small
\[
Q_k = \left[ \begin{array}{ccccc}
1       & -q    & 0     & \ldots        &0\\
0       & \ddots        & \ddots        & \ddots        &\vdots \\
\vdots  & \ddots        & \ddots        & \ddots        &0\\
\vdots  &       & \ddots        & \ddots        &-q\\
0       & \ldots        & \ldots        & 0     &1
\end{array}\right] 
.
\]}

\subsection{Tabulations of the Matrix-Vector Product and Time-Memory Trade-off}
Tabulating fully the multiplication by $Q_k$
requires a table of size at
least $p^{k+1}$. However, the matrices $Q_k$ can be decomposed recursively into smaller similar matrices as follows:

\centerline{\setlength{\unitlength}{2144sp}%
\begin{picture}(2453,1478)(572,-1833)
\thinlines
\put(2945,-381){\line( 1, 0){ 68}}
\put(3013,-381){\line( 0,-1){1440}}
\put(3013,-1821){\line(-1, 0){ 68}}
\put(1689,-367){\line(-1, 0){ 68}}
\put(1621,-367){\line( 0,-1){1440}}
\put(1621,-1807){\line( 1, 0){ 68}}
\put(1666,-1186){\framebox(765,765){}}
\put(2259,-1155){\makebox(0,0)[lb]{1}}
\put(1863,-1552){\makebox(0,0)[lb]{0}}
\put(2626,-755){\makebox(0,0)[lb]{0}}
\put(1244,-1140){\makebox(0,0)[lb]{=}}
\put(587,-1178){\makebox(0,0)[lb]{$Q_{i+j}$}}
\put(1735,-771){\makebox(0,0)[lb]{$Q_i$}}
\put(2460,-1533){\makebox(0,0)[lb]{$Q_j$}}
\put(2206,-1726){\framebox(765,765){}}
\end{picture}%
}
It is therefore easy to tabulate the product by~$Q_k$ with an adjustable table of size $p^j$.
\begin{prop}Let $q$ be a power of~2. Given a table of size~$p^j$ ($1\le j\le k+1$), a REDQ${}_k$CORRECTION can be performed with $\lfloor(k-1)/(j-1)\rfloor$ table accesses.
\end{prop}
When $q$ is a power of $2$,
the computation of the $u_i$ in the first part of Algorithm~\ref{alg:REDQ} requires $1~\text{div}$ and $({k+1})/{2}~\text{axpy}$
as shown in Proposition~\ref{prop:redqcost}.
Now, the time memory trade-off enables to compute
the second part efficiently.
%

\begin{exmp}\label{ex:tmto}
We compute the corrections for a degree $6$ polynomial. One can
tabulate the multiplication by $Q_6$, a $7\times 7$ matrix, or
actually, by only the first $6$ rows of $Q_6$, denoted by
$\underline{Q}_6$,
with therefore $p^7$ entries, each of size $6\log_2(p)$. Instead, one can
tabulate the multiplication by $\underline{Q}_2$, a $2\times 3$
matrix. To compute
$[\mu_0,\ldots,\mu_6]^T = Q_6 [u_0\ldots,u_6]^T = [\underline{Q}_6[u_0\ldots,u_6]^T,u_6] $ it is sufficient to use
three multiplications by $\underline{Q}_2$ as shown in the
following algorithm:
\begin{algorithm}[H]
\caption{$Q_6$ with an extra memory of size $p^3$}
\begin{algorithmic}[1]
\REQUIRE $[u_0\ldots,u_6] \in (\pF{p})^7$;
\REQUIRE a table $\underline{Q}_2$ of the associated $2\times 3$ matrix-vector
multiplication over $\pF{p}$.
\ENSURE $[\mu_0,\ldots,\mu_6]^T = Q_6 [u_0\ldots,u_6]^T$.
\STATE $a_0,a_1 = \underline{Q}_2 [u_0,u_1,u_2]^T$;
\STATE $b_0,b_1 = \underline{Q}_2 [u_2,u_3,u_4]^T$;
\STATE $c_0,c_1 = \underline{Q}_2 [u_4,u_5,u_6]^T$;
\STATE Return $[\mu_0,\ldots,\mu_6]^T= [a_0,a_1,b_0,b_1,c_0,c_1,u_6]^T$;
\end{algorithmic}
\end{algorithm}
\end{exmp}

\paragraph*{Note on Indexing.}
In practice, indexing by a tuple of integers mod $p$
is made by evaluating at $p$, as $\sum u_i p^i$. If more memory
is available, one can also directly index in the binary format
using $\sum u_i \left(2^{\lceil \log_2(p) \rceil}\right)^i$. On the
one hand all
the multiplications by $p$ are replaced by bit fields extraction. On the other
hand, this makes the table grow a little bit, from $p^k$ to
$2^{\lceil \log_2(p) \rceil k}$.

\medskip

In the rest of this article we
restrict to the case when $q$ is a power of two.



\section{Euclidean Division by Floating Point Routines}\label{sec:floor}
In the implementation of REDQ above, there remains one machine
division (Line~1). 
Since exact division is rather time-consuming on modern
processors, it is natural to try and replace it by floating point operations. 
If $r$ and $p$ are integers and
we want to compute $r/p$, 
then computing $r/p$ by a floating point
division with a rounding to nearest mode followed
by flooring produces the expected value~\cite[Theorem
1]{Lefevre:2005:div}.

Instead of a division, a multiplication
by a precomputed inverse of $p$ can be used (this is done e.g., in
NTL \citep[Theorem 1.5]{Shoup:2005:NTB}). However,
 in that case, the correctness is not guaranteed for all
representable $r$ (see \cite{Lefevre:2005:mul} for details). Thus, as
is done e.g. in
NTL, one has to make some additional tests and corrections. 
In this section we propose bounds on $r$ for which this correctness
is guaranteed, taking the rounding modes into account. Outside of these bounds, we show 
%
that a difference of (at most) 1 after the flooring is possible for some values
of $r$, $p$ and the rounding modes. This can be detected by
two tests on
the resulting residue (below $0$ or above $p$) as is done by
\cite{Shoup:2007:NTL}. We show that only one of these tests is
mandatory if the rounding modes can be mixed. 
The inverse of the prime can be
precomputed for each mode, which avoids the need for costly changes of modes.
This is summarized in Algorithm~\ref{alg:AFDIV} at the end of this section.

\subsection{Rounding Modes}
We assume that the floating-point arithmetic follows the IEEE~754
standard and we denote by $\operatorname{ulp}(z)$ the \emph{unit in the last place}
of $z$ for a floating-point number $z$ such that
\[ 2^{\beta-1} \le |z| \le 2^\beta - 1.\]
For each operation three rounding modes are possible (\emph{up}
$\bigtriangleup(\cdot)$, \emph{down} $\bigtriangledown(\cdot)$ and
\emph{nearest}\footnote{How ties are broken is irrelevant here.}
$\diamond(\cdot)$). If two computations take place at different times, each of them can be performed in a different rounding mode, so that we have $9$ cases to consider. 
It is worth
investigating all cases since changing rounding modes is a costly
operation and the FPU may be in a particular rounding mode due to
constraints in the surrounding code, or some rounding modes are simply
not available in the particular computing environment. 
Also, computing the best bounds enables to make the best of 
delayed reduction in modular computations.

\subsection{Floating Point Division}
Denoting by $r=kp+u$ the Euclidean division of $r$ by $p$ where $0\le
r \le 2^\beta -1$ we are interested to know under which conditions on
$r$ and $p$ Algorithm~\ref{alg:FDIV} returns the quotient~$k$, depending on the rounding modes $\circ_1$ and $\circ_2$.
\begin{algorithm}[ht]
\caption{FDIV}\label{alg:FDIV}
\begin{algorithmic}[1]
\REQUIRE One integer $r$ such that $0 \le r \le 2^\beta - 1$;
\REQUIRE one integer $p$ such that $1 \le p \le 2^\beta-1$; 
\REQUIRE two choices of rounding-modes $\circ_1$ and $\circ_2$.
\ENSURE $\lfloor \frac{r}{p} \rfloor$.
\STATE $invp \gets \circ_1(1/p)$
\STATE $x \gets \circ_2(r\cdot invp)$
\STATE Return $\lfloor x \rfloor$.
\end{algorithmic}
\end{algorithm}
\newcommand{\ulp}{\text{ulp}}

\subsection{Results}
Our results are summarized in Table~\ref{tab:fdiv}. 

The
column ``Range'' gives guaranteed bounds on the result. This shows which
tests and corrections may be needed in order to compute the expected
value $k$. The interval given in this column is optimal, in the sense that we have
examples where $\lfloor x \rfloor \neq k$ in each possible direction.

The column ``Bound on $r$'' gives a strict upper bound on $r$ under which
$\lfloor x \rfloor \le k$, in those cases where 
the result could indeed overflow. We do not know whether these bounds
are optimal; for some cases we could find a systematic family
of examples indexed by $\beta$ that reach the bounds asymptotically; 
other bounds could be approached by exhaustive search on small
$\beta$. These examples are not included here. We believe that
all bounds are optimal except for case 2 where a value closer to
$\frac{3}{8}$ (instead of $\frac{1}{3}$) could be
reached (take $\beta=2n+1$ and $p=2^{n}-1$,
$r= (3\cdot 2^{n-2}+3)p-1$). For our purpose
$\frac{1}{3}$ is close enough.

The column ``Lost bits'' gives a
simplified version of this bound: if $r$ fits in $\beta$ minus this number
of bits, then $\lfloor x \rfloor \le k$.

%
\begin{table}[ht]
\centering
\begin{tabular}{|c|c|c|c|c|c|}
\hline
Case & $\circ_1$ & $\circ_2$ & Range & Bound on $r$ & Lost bits\\
\hline
1    & $\bigtriangleup(\cdot)$ & $\bigtriangleup(\cdot)$ & $k \le \lfloor x
\rfloor \le k+1$ & $2^\beta/({4+2^{2-\beta}})$ & 3 \\
\hline
2    & $\bigtriangleup(\cdot)$ & $\diamond(\cdot)$ & $k \le \lfloor x
\rfloor \le k+1$ & $2^\beta/({3+2^{1-\beta}})$ & 2 \\
\hline
3    & $\bigtriangleup(\cdot)$ & $\bigtriangledown(\cdot)$ & $k \le \lfloor x
\rfloor \le k+1$ & $2^\beta/2$ & 1 \\
\hline
4    & $\diamond(\cdot)$ & $\bigtriangleup(\cdot)$ & $k \le \lfloor x
\rfloor \le k+1$ & $2^\beta/({3+2^{1-\beta}})$ & 2 \\
\hline
5    & $\diamond(\cdot)$ & $\diamond(\cdot)$ & $k-1 \le \lfloor x
\rfloor \le k+1$ & $2^{\beta-1}/({1+2^{-1-\beta}})$ & 2 \\
\hline
6    & $\diamond(\cdot)$ & $\bigtriangledown(\cdot)$ &  $k-1 \le \lfloor x
\rfloor \le k$ & -- & 0 \\
\hline
7    & $\bigtriangledown(\cdot)$ & $\bigtriangleup(\cdot)$ & $k-1 \le \lfloor x
\rfloor \le k+1$ & $2^\beta/2$ & 1 \\
\hline
8    & $\bigtriangledown(\cdot)$ & $\diamond(\cdot)$ & $k-1 \le \lfloor x
\rfloor \le k$ & -- & 0 \\
\hline
9    & $\bigtriangledown(\cdot)$ & $\bigtriangledown(\cdot)$ & $k-1 \le \lfloor x
\rfloor \le k$ & -- & 0 \\
\hline
\end{tabular}\\[1mm]
\caption{Possible values of $\lfloor x\rfloor$ and bounds on $r$ such that~$\lfloor x\rfloor\le k$}
\label{tab:fdiv}
\end{table}
\subsection{Proof of the Bounds on~$\lfloor x\rfloor$}
We denote by $\epsilon_1$ and $\epsilon_2$ the errors in rounding
as follows:
\[
invp	 = 	\frac{1}{p}(1 + \epsilon_1), \qquad
x	 = 	(r\cdot invp)(1 + \epsilon_2).
\]
Thus the value that is computed is
\begin{align}
 x   & = 	\frac{r}{p}(1+\epsilon_1)(1 + \epsilon_2) \nonumber \\
	& = 	k + \frac{u}{p} +(\epsilon_1 + \epsilon_2 +\epsilon_1\epsilon_2)\frac{r}{p} =:   k + R
	\label{eq:def_R}
\end{align}
where $R$ is the term we want to bound. For example $R<1$ means $\lfloor x \rfloor
\le k$.
Bounds on $\epsilon_1$ and
$\epsilon_2$ depend on the rounding mode, but in all cases we have
\begin{equation}\label{ineq:epsilon}
|\epsilon_i|\le
2^{1-\beta},\qquad i\in\{1,2\}.
\end{equation}
              
\begin{lemma}The result of Algorithm~\ref{alg:FDIV} is never off by more
than one unit, that is
\[ k-1 \le \lfloor x \rfloor \le k + 1.\]
\end{lemma}
\begin{proof}
First we show $R < 2$, which gives the upper bound.
This is obtained by injecting the inequalities~\eqref{ineq:epsilon}, in the definition~\eqref{eq:def_R} of~$R$:
\begin{align}
    R & \le  \frac{p-1}{p} + (2^{2-\beta}+2^{2-2\beta})\frac{r}{p} \label{ineq:R}\\
      & \le  1 -\frac{1}{p}+  (2^{2-\beta}+2^{2-2\beta})\frac{2^\beta-1}{p} \nonumber\\
      & =  1 - \frac{1}{p}(3 - 2^{2-2\beta}).\nonumber
\end{align}
Thus $R<2$ for $p\ge 3$. In the special case $p=2$, we have $\epsilon_1=0$
so that the result still holds.

Similarly, in the other direction we have
\begin{align*}
    R & \ge  -(2^{2-\beta} + 2^{2-2\beta})\frac{r}{p}\\
      & \ge  -\frac{1}{p}(4-2^{2-2\beta})
\end{align*}
and $R\ge -1$ for $p\ge 4$. For $p=2$, the results follows from $\epsilon_1=0$. For the last case, $p=3$, we first analyze more precisely the rounding error~$\epsilon_1$. The binary expansion of
$\frac{1}{3}=0.01010101010\ldots$ implies that
\[ \circ(\frac{1}{3}) = 
    \begin{cases}
    \frac{1}{3}(1+2^{-\beta-1}) & \text{if $\beta$ is even and }\circ(\cdot)\in\{\bigtriangleup(\cdot), \diamond(\cdot)\},\\
    \frac{1}{3}(1+2^{-\beta-1}) & \text{if $\beta$ is odd and }\circ(\cdot)\in\{\diamond(\cdot), \bigtriangledown(\cdot)\},\\
    \frac{1}{3}(1+2^{-\beta}) & \text{otherwise.}
    \end{cases}
\]
Thus in all cases $|\epsilon_1|\le 2^{-\beta}$. Using this better bound in the computation above completes the proof.
\end{proof}

\begin{lemma}In cases 1, 2, 3 and 4, $\lfloor x\rfloor\ge k$.
\end{lemma}
\begin{proof}
In the first three cases, $invp$ is rounded up, and thus
\[    r\cdot invp  =  k(p\cdot invp) + u\cdot invp\ge k.\]
This implies that $\circ(r\cdot invp)  \ge k$
because rounding modes are monotone and $k$ is exactly representable.

In case~4, since $|\epsilon_1| \le 2^{-\beta}$, we have
\[ invp > (1 - 2^{-\beta})\frac{1}{p}. \]
Then
\[
r \cdot invp  =  (kp+u)invp
	      \ge  kp\cdot{}invp 
	      >  k(1-2^{-\beta}).
\]
Again, $k$ is an integer thus exactly representable. Denote by $k^-$ the largest floating-point
number that is strictly less than $k$:
\[
k^- = \begin{cases}
    k - \frac{1}{2}\ulp(k) & \text{if $k$ is a power of 2,}\\
    k - \ulp(k) & \text{otherwise.}
    \end{cases}
\]
If $k$ is a power of $2$, then the way of rounding
$k(1-2^{-\beta})$ is the same as that of $1-2^{-\beta}$ since $k$ only
changes the exponent in the result. Since
$\bigtriangleup(1-2^{-\beta}+\delta) \ge 1$ for any $\delta > 0$ we
get $x = \bigtriangleup(r\cdot invp) \ge k$.
If $k$ is not a power of $2$, then
\[
    k^-  =  k - \ulp(k)  <  k(1 - 2^{-\beta})<
    r\cdot invp\]
and thus again $x=\bigtriangleup(r\cdot invp) \ge  k$.
\end{proof}

\begin{lemma}In cases 6, 8 and 9, $\lfloor x\rfloor\le k$.
\end{lemma}
\begin{proof}
In case 6, we have $|\epsilon_1|\le 2^{-\beta}$. This implies
\[r\cdot{}invp  =  \frac{r}{p}(1+\epsilon_1)
\le \frac{r}{p}(1+2^{-\beta})
 <  \frac{r}{p} + \frac{1}{p}<k+1\]
and therefore $x  =  \bigtriangledown(r\cdot{}invp) < k+1$, whence 
$\lfloor x \rfloor  \le  k.$

In cases 8 and 9, since $invp =
\bigtriangledown(\frac{1}{p})$, we have
\begin{equation}\label{ineq:case8}
 r \cdot invp \le \frac{r}{p} \le k+1-\frac{1}{p}.
\end{equation}
In case~9, this is sufficient to conclude. Otherwise, as in case~4 above, denote by $(k+1)^-$ the floating-point number that is just below
$k+1$ 
and let $m$ be the midpoint between $k+1$ and $(k+1)^-$.  If we show
$r\cdot{}invp < m$, then $\diamond(r\cdot{}invp)\le (k+1)^-$ and
therefore $\lfloor x \rfloor \le k$.

First assume that $k+1$ is a power of $2$. In this case
$\ulp(k+1)=2^{1-\beta}(k+1)$ and
$m=k+1-\frac{1}{4}\ulp(k+1)=(k+1)(1-2^{-\beta-1})$.
Since by assumption $p  <  2^{\beta} $ and $r< 2^{\beta}$, 
so $p(k+1) \le r+p < 2^{\beta+1}$ and it follows that
\[    \frac{1}{p} >  (k+1)2^{-\beta-1}. \]
Together with~\eqref{ineq:case8} this gives
\[r\cdot{}invp <  (k+1)(1-2^{-\beta-1})=m.\]

Assume now that $k+1$ is not a power of $2$ (and $k\neq 0$). Then
$\ulp(k+1)=\ulp(k)\le 2^{1-\beta}k$, $m=k+1-\frac{1}{2}\ulp(k)$ and
$\frac{1}{p} >  k2^{-\beta} \ge \ulp(k)/2$. Thus finally,
\[   r\cdot{}invp  <    k+1-\frac{1}{2}\ulp(k)\le m.\]
If $k=0$ then $\frac{1}{2}\ulp(k+1)=2^{-\beta} < \frac{1}{p}$ and the
result follows.
\end{proof}

\subsection{Bounds on~$r$ making the result exact}
In cases~6, 8 and 9, the result of Algorithm~\ref{alg:FDIV} is smaller than~$k$ for all values of~$r$.
We now proceed to a case-by-case proof of each remaining row of Table~\ref{tab:fdiv}, giving bounds on~$r$ such that this happens.

\paragraph*{Case 1.} In this case $0 \le \epsilon_i \le 2^{1-\beta}$
for $i\in\{1,2\}$
and thus \eqref{ineq:R} becomes
\[ R \le \frac{p-1}{p} + (2^{2-\beta} + 2^{2-2\beta})\frac{r}{p}\]
so that $|R|<1$ is implied by
\[ r < 2^\beta\frac{1}{4+2^{2-\beta}} \]
that is close to ${2^\beta}/{4}$ and we lose less than 3 bits compared to
the bound $r<2^\beta$ required to have no loss of precision
on $r$.

\paragraph*{Case 2.}
In this case $|\epsilon_2| \le 2^{-\beta}$ and
\eqref{ineq:R} becomes
\[ R \le \frac{p-1}{p} + (3\cdot 2^{-\beta}+2^{1-2\beta})\frac{r}{p}
\]
The condition $R<1$ is implied by
\[ 3\cdot 2^{-\beta} + 2^{1-2\beta} < \frac{1}{r} \]
that is $r < 2^\beta/({3+2^{1-\beta}})$ which is close to
${2^\beta}/{3}$, and less that two bits are lost. 

\paragraph*{Case 3.}
In this case $-2^{1-\beta}\le \epsilon_2 \le 0$,
$|\epsilon_1 + \epsilon_2| \le 2^{1-\beta}$ and $\epsilon_1\epsilon_2 \le 0$.
We get
\[ R \le 1 - \frac{1}{p} + \frac{r}{p}2^{1-\beta} \]
and the condition $R<1$ is ensured by
\[ r < \frac{1}{2}2^\beta \]
which means we lose one bit.

\paragraph*{Case 4.} 
This is as in
case 2 since $\bigtriangleup(\cdot)$ and $\diamond(\cdot)$ play a
symmetric role in the analysis.

At this stage, we have obtained the following.
\begin{prop}In Cases 1--4, Algorithm~\ref{alg:FDIV} is correct for $r$ obeying the bounds of Table~\ref{tab:fdiv}.
\end{prop}
The remaining cases are proved similarly:
\paragraph*{Case 5.}
We have $|\epsilon_1|\le 2^{-\beta}$, $|\epsilon_2|\le 2^{-\beta}$
and $|\epsilon_1+\epsilon_2+\epsilon_1\epsilon_2|\le 2^{1-\beta}+2^{-2\beta}$.
Then~\eqref{ineq:R} becomes
\[R \le 1 + \frac{1}{p}(1-2^{-\beta}-2^{-2\beta}).\]
The condition $R <1$ is implied by
$ r < 2^{\beta-1}/({1+2^{-1-\beta}})$
which is close to $\frac{1}{2}2^\beta$.

\paragraph*{Case 7.} The  bound on $r$ follows from case 3 as
$\epsilon_1$ and $\epsilon_2$ play a symmetric role in the error
analysis of case 3.

\subsection{Using Algorithm \ref{alg:FDIV}}
\begin{algorithm}[ht]
\caption{Applied FDIV}\label{alg:AFDIV}
\begin{algorithmic}[1]
\REQUIRE $r$, integer such that $0 \le r \le 2^\beta - 1$;
\REQUIRE $p$, integer such that $1 \le p \le 2^\beta-1$; 
\ENSURE $\lfloor \frac{r}{p} \rfloor$.
\vspace{5pt}\newline{\underline{Constants}}\vspace{2pt}
\STATE $B_\bigtriangleup \gets 2^\beta/({3+2^{1-\beta}})$
\STATE $B_\diamond \gets 2^\beta/({3+2^{1-\beta}})$
\STATE $B_\bigtriangledown \gets 2^{\beta-1}$
\vspace{5pt}\newline{\underline{Precomputation}}\vspace{2pt}
\STATE $invp_\bigtriangleup \gets \diamond(1/p)$
\STATE $invp_\diamond \gets \bigtriangleup(1/p)$
\STATE $invp_\bigtriangledown \gets \bigtriangleup(1/p)$
\vspace{5pt}\newline{\underline{Division}}\vspace{2pt}
    \STATE $x \gets  \circ(r \cdot invp_\circ) $
\STATE $y \gets \lfloor x \rfloor$
\vspace{5pt}\newline{\underline{Possible correction}}\vspace{2pt}
\IF{$r \ge B_\circ$} \label{alg:AFDIV:correction}
    \STATE $z \gets p\cdot y$
    \IF{$z > r$}
	\STATE $y \gets y - 1$
    \ENDIF
\ENDIF
\STATE Return $y$.
\end{algorithmic}
\end{algorithm}

Algorithm \ref{alg:AFDIV} demonstrates how to apply the results of
Table \ref{tab:fdiv} in a program. We precompute $1/p$ in several
rounding modes and use the best version in the subsequent
multiplication, depending on the current rounding mode $\circ(\cdot)$.
The strategy used here is to make sure $\lfloor x \rfloor \ge k$ after
the multiplication so only one test at most is needed for the
correction. In addition we choose the version that maximizes the bound
$B$.

It should be noted that there is no strategy in the choice of
$\circ_1(\cdot)$ that guarantees $\lfloor x \rfloor \le k$ for any
choice of $\circ_2(\cdot)$ (take $\circ_2(\cdot) =
\bigtriangleup(\cdot)$), meaning that the described strategy is indeed
the only one than minimizes the number of tests needed for the
correction.

In a typical application of Algorithm \ref{alg:AFDIV} where $r$ is the
accumulation of several products modulo $p$, the bound $B$ can be
interpreted in the numbers of operations that can be performed before
a reduction is necessary. If this number is not exceeded then the
correction is never needed.

%


\section{Application 1: Polynomial Multiplication}\label{sec:del}
\subsection{Delayed Reduction}
A classical technique 
for modular polynomial multiplication is to
use only delayed reductions. The idea is to compute each polynomial
coefficient by a delayed dotproduct e.g., as in \citep{jgd:2004:dotprod}:
products of the form $\sum_i a_i b_{k-i}$ are accumulated
without reductions, and the overflow is dealt with in one final pass. Thus, with a centered representation modulo $p$ for instance (integers from $({1-p})/{2}$ to
$({p-1})/{2}$), it is possible to accumulate at least $n_d$ products as long as
\begin{equation}\label{eq:delayed}
n_d (p-1)^2 < 2^{\beta+1}.
\end{equation}
The final modular reduction can be performed in many different ways
(e.g., classical division, floating point multiplication by the
inverse, Montgomery reduction, etc.), we just call the best one REDC
here. At worst, it is equivalent to 1~machine division.

\subsection{Fast Q-adic Transform}
We represent modular polynomials
of the form $P = \sum_{i=0}^N a_i X^i$ by
$P = \sum P_i \left( X^{d+1} \right) ^{i}$ where the $P_i$'s are 
polynomials of degree $d$ stored in a single integer in the $q$-adic way.

Then a product $PQ$ has the form
$$PQ=\sum \left(\sum P_i Q_{t-i}\right) \left(X^{d+1}\right)^{t},$$
where each multiplication $P_i Q_{t-i}$ is computed by Algorithm~\ref{alg:dqt} on a single machine integer.
The final reduction is performed by a tabulated REDQ and
can also be delayed as long as conditions
(\ref{eq:bounds}) are guaranteed.

\subsection{Comparison}
We use the following
complexity model:
multiplications and additions in the field are counted as
atomic operations while the machine divisions are counted separately. For instance,
we approximate REDC by one machine division and an axpy. 
We recall that REDQ$_k$ denotes a
simultaneous reduction of $k$ residues. 
In this complexity model a
REDQ$_k$ thus requires $1$ division and $k/2$ multiplications and
additions. A REDQ with $k$ residues is a Kronecker substitution with a
polynomial of degree $d=k-1$.
We also call $k$-FQT the use of a polynomial of degree
$d=k-1$ for the $q$-adic substitution. 
Thus a multiplication $P_i Q_{t-i}$ in a
$k$-FQT requires the reduction of $2d+1$ coefficients, i.e., a
REDQ$_{2d+1}=$REDQ$_{2k-1}$.

Let $P$ be a polynomial of degree $N$ with indeterminate $X$.
If we use a $(d+1)$-FQT, it will then become a polynomial of degree $D_q$
in the indeterminate $Y=X^{d+1}$, with
$$D_q = \left\lceil \frac{N+1}{d+1} \right\rceil -1.$$
Table~\ref{tab:mulpol} gives the respective
complexities of both strategies. The values of~$n_d$ and~$n_q$ are given by Equations~(\ref{eq:bounds},\ref{eq:delayed}):
\[n_d = \left\lfloor\frac{2^{\beta+1}}{(p-1)^2} \right\rfloor,\quad
n_q = \left\lfloor \frac{q}{(d+1)(p-1)^2} \right\rfloor~\text{with}~q=2^{\frac{m}{2d+1}}.\]

\begin{table}[ht]\center
\begin{tabular}{|l||c|c|}
\hline
& Mul \& Add & Reductions \\
\hline
Delayed & $(2N+1)^2$ & $(2N+1)\left\lceil \frac{2N+1}{n_d} \right\rceil$ REDC\\
d-\fqt    & $(2D_q+1)^2$ & $(2D_q+1)\left\lceil\frac{2D_q+1}{n_q}\right\rceil$ REDQ$_{2d+1}$ \\
\hline
\end{tabular}
\caption{Modular polynomial multiplication complexities.}\label{tab:mulpol}
\end{table}
\begin{exmp}
With $p=3$, $N=500$, $n_d$ is much larger than $2N+1$ and thus the classical delayed polynomial multiplication algorithm requires
$10^6$ multiplications and additions and 
$10^3$ remainderings. 

If we choose a double floating point
representation and a 4-FQT 
(i.e 4 coefficients in a word, or a degree 3 substitution), the fully
tabulated \fqt boils down
to $8.6\cdot 10^4$ multiplications and additions and $5.7\cdot 10^3$
divisions. On the one hand, the number of operations is therefore
reduced by a factor close to~11. On the other hand the delayed
 code can compute every coefficient with a
single reduction in this case, 
while the FQT code has to compute less coefficients, 
but breaks the pipeline.
\end{exmp}
\begin{exmp}
Even by switching to a larger mantissa, say e.g., 128 bits, so that the
\dqt multiplications are roughly 4 times as costly as double floating
point operations, the FQT can still be useful.

Taking $p=1009$ and choose $d=2$, this still gives around
$1.3\cdot 10^5$ multiplications and additions over 128
bits and $7\cdot 10^3$ divisions. The number of operations is still
reduced by a factor of~7. 
This should therefore still be faster than the
delayed multiplication over $32$ bits.
\end{exmp}

\begin{figure}[ht]\center
\includegraphics[width=\textwidth*6/13]{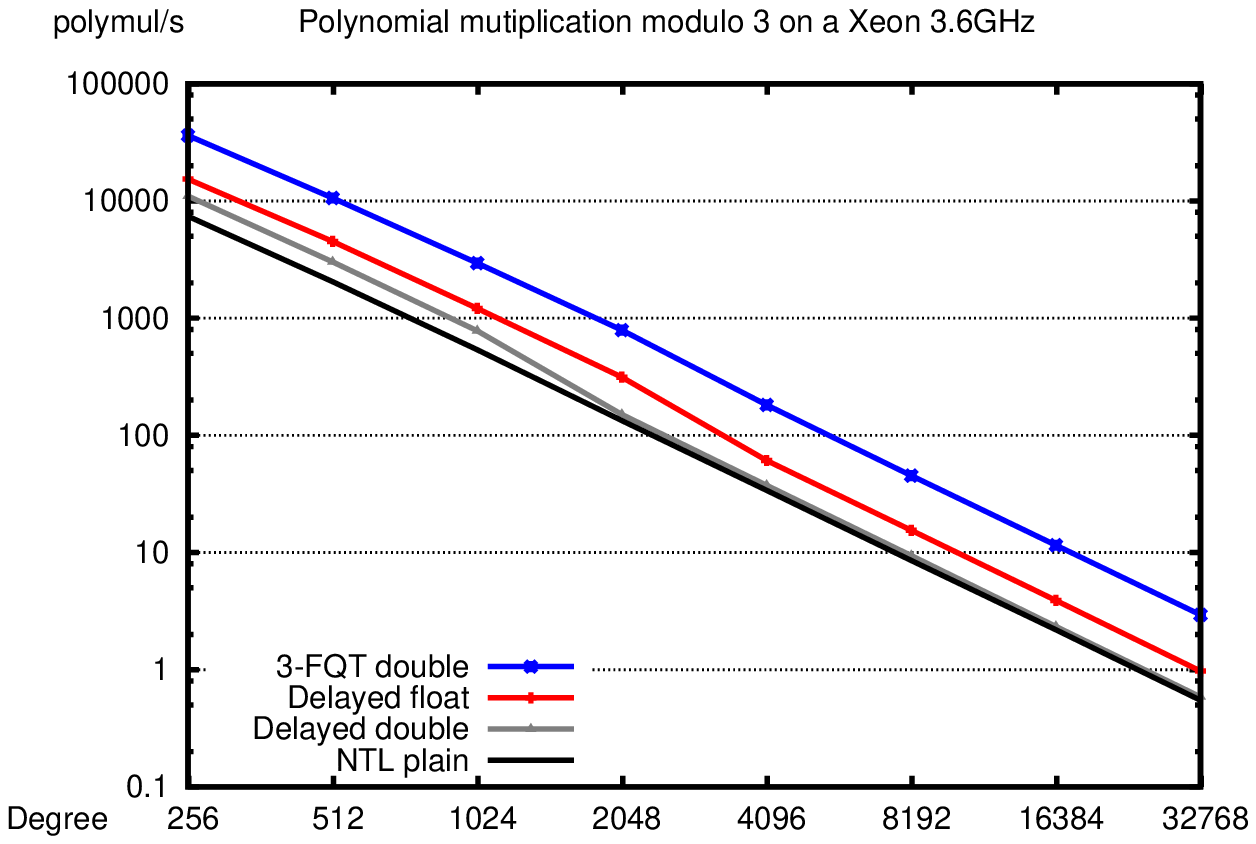}
\hfill\includegraphics[width=\textwidth*6/13]{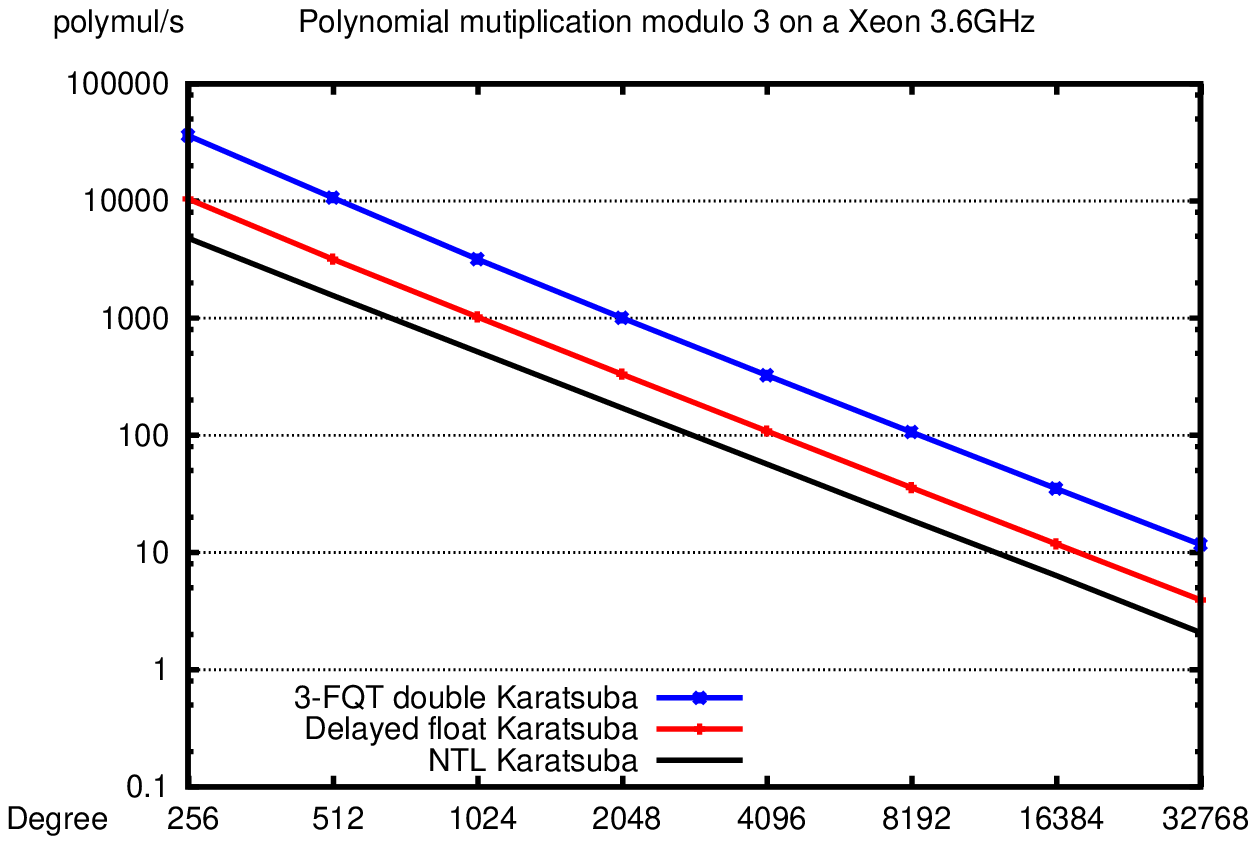}
\caption{Number of classical and Karatsuba polynomial multiplications modulo 3 per second on a Xeon
3.6 GHz (logarithmic scale)}\label{fig:pmkara}
\end{figure}


On Figure~\ref{fig:pmkara}, 
we compare our two implementations with that of NTL
\citep{Shoup:2007:NTL}. We see that
the \fqt is faster than NTL as long as the same algorithm is 
used. This shows that our strategy is very useful for small degrees and small primes;
not only for the classical algorithm (left) but also for subquadratic ones 
(right): the use of \fqt leads to a gain of 
an order of magnitude.
Note however that in the special case~$p=2$, NTL offers a very optimized 
implementation which is still an order of magnitude faster than our
general purpose implementation: specific binary routines, such as the
ones proposed by \citep{Weimerskirch:2003:gftwo}, enables to pack
coefficients as bits of machine words. 


\section{Application 2: Small Finite Field Extensions}\label{sec:fflas}
The isomorphism between finite fields of equal sizes gives  a canonical
representation:
any finite
field extension is viewed as the set of polynomials modulo a prime $p$
and modulo an
irreducible polynomial ${\cal P}$ of degree $k$.
Clearly we can thus convert any finite field element to its $q$-adic
expansion; perform the \fqt between two elements and then reduce the
polynomial thus obtained modulo ${\cal P}$.
Furthermore, it is possible to use floating point routines to perform
exact linear algebra as demonstrated by \cite{jgd:2009:toms}.

We use the strategy of \cite[Algorithm 4.1]{jgd:2002:fflas}:
convert
vectors over $\GF{p^k}$ to $q$-adic floating point; call a fast
numerical linear algebra routine (BLAS); convert the
floating point result back to the usual field representation.
We improve all the conversion steps
as follows:
\begin{enumerate}
\item replace the Horner evaluation of the polynomials, to form the
  $q$-adic expansion, by a single table look-up, recovering directly the
  floating point representation;
\item replace the radix conversion and the costly modular reductions
  of each polynomial coefficient, by a single REDQ operation;
\item replace the polynomial division by two table look-ups and
a single field operation.
\end{enumerate}

\begin{algorithm}[ht]
\caption{Fast Dot product over Galois fields via \fqt and \fqt inverse}
\label{alg:FGDP}
\begin{algorithmic}[1]
\REQUIRE A field $\GF{p^k}$ with elements represented as exponents of
a generator of the field;
\REQUIRE two vectors $v_1$ and $v_2$ of elements of $\GF{p^k}$;
\REQUIRE a sufficiently large integer $q$.
\ENSURE $R \in \GF{p^k}$, with $R = v_1^T\cdot v_2$.
\vspace{5pt}\newline{\underline{Tabulated $q$-adic conversion}}\vspace{2pt}
\newline\COMMENT{Use conversion tables from exponent to floating point evaluation}
\STATE Set $\widetilde{v_1}$ and $\widetilde{v_2}$ to the floating point vectors
of the evaluations at $q$ of the elements of $v_1$ and $v_2$.
\vspace{5pt}\newline{\underline{The floating point computation}}\vspace{2pt}
\STATE Compute $\tilde{r} = \widetilde{v_1}^T \cdot\widetilde{v_2}$;
\vspace{5pt}\newline{\underline{Computing a radix
    decomposition}}\vspace{2pt}
\STATE $r = REDQ\_COMPRESSION(\tilde{r},p,q)$;
\vspace{5pt}\newline{\underline{Variant of REDQ\_CORRECTION}}\vspace{2pt}
\newline\COMMENT{$\mu_i$ is such that $\mu_i = \widetilde{\mu_i} \bmod p$ for
$\tilde{r} = \sum_{i=0}^{2k-2} \widetilde{\mu_i} q^i$}
\STATE Set $L = representation( \sum_{i=0}^{k-2} \mu_i X^i )$.
\STATE Set $H = representation( X^{k-1} \times \sum_{i=k-1}^{2k-2} \mu_i X^{i-k+1} )$.
\vspace{5pt}\newline{\underline{Reduction in the field}}\vspace{2pt}
\STATE Return $R = H + L \in \GF{p^k}$;
\end{algorithmic}
\end{algorithm}
This is presented in Algorithm~\ref{alg:FGDP}.
Line~1 is
the table look-up of floating point values associated to elements of
the field; line 2 is the numerical computation; line 3 is the
first part of the REDQ reduction; lines 4 and 5 are a time-memory
trade-off with two table accesses for the corrections of REDQ, combined
with a conversion from polynomials to discrete logarithm
representation; the last line 6 combines the latter two results, inside
the field.
A variant of REDQ is used in Algorithm~\ref{alg:FGDP},
but $u_i$ still satisfies $u_i =
\sum_{j=i}^{2k-2} \mu_j q^{j-i} \bmod p$ as shown in
Theorem~\ref{thm:REDQ}. Therefore the representations of
$\sum \mu_i X^j$ in the field can be precomputed and stored in two tables
where the indexing will be made by $(u_0,\ldots,u_{k-1})$ and
$(u_{k-1},\ldots,u_{2k-2})$ and not by the $\mu_i$'s.

Thus, this algorithm approaches the performance of the prime
field wrapping also for small extension fields. Indeed,
suppose the internal representation of the extension field is already by
discrete logarithms and uses conversion tables from polynomial to index
representations (see e.g., \cite{jgd:2004:dotprod} for details).
Then we choose a time-memory trade-off for the REDQ operation of the
same order of magnitude, that is to say $p^k$.
The overall memory required by these new tables only
doubles and the REDQ requires only $2$ accesses.
Moreover, in the small extension, the polynomial multiplication must
also be reduced by an irreducible polynomial, ${\cal P}$.
This reduction can be precomputed
in the REDQ table look-up and is therefore almost free.
Moreover, many things can be factorized if the field representation is
by discrete logarithms. For instance, the elements are represented by their discrete
logarithm with respect to a generator of the field, instead of 
polynomials. In this case
there are already some table accesses
for many arithmetic operations, see
e.g., \citep[\S 2.4]{jgd:2004:dotprod} for details.

%
\begin{thm}\label{thm:fgdp}
Algorithm~\ref{alg:FGDP} is correct.
\end{thm}
\begin{pf}
We have to prove that it is possible to compute $L$ and $H$ from
the $u_i$'s. We have $\mu_{2k-2}=u_{2k-2}$ and
$\mu_i = u_i-qu_{i+1} \bmod p$, for $i=0,\dots,2k-3$.
Therefore a precomputed table of $p^k$ entries, indexed
by $(u_0,\ldots,u_{k-1})$, can provide the representation of
$$L=\sum_{i=0}^{k-2} (u_{i}-qu_{i+1} \bmod p) X^i.$$
Another table with $p^k$ entries, indexed by
$(u_{k-1},\ldots,u_{2k-2})$, can provide the representation of
$$H=u_{2k-2}X^{2k-2}+\sum_{i=k-1}^{2k-3} (u_{i}-q u_{i+1} \bmod p) X^{i}.$$

Finally $R = X^{k-1} \times \sum_{i={k-1}}^{2k-2} \mu_i X^{i-k+1}  +
\sum_{i=0}^{k-2} \mu_i X^i $ needs to be reduced modulo the
irreducible
polynomial used to build the field. But, if we are given the
representations of $H$ and $L$ in the field, $R$ is then equal
to their sum inside the field, directly using the internal
representations.
\end{pf}

Table~\ref{tab:comp} recalls the respective complexities of the conversion
phase in both algorithms. Here, $q$ is a power of two
and the REDQ division is computed via the floating point routines
of Section~\ref{sec:floor}.
\begin{table}[ht]
\begin{center}
  \begin{tabular}[ht]{|l||c|c|c|}
\hline
   & Alg.~\ref{alg:dqt}& Alg.~\ref{alg:FGDP} \\
Memory  & $3p^k$ & $4p^k+2^{1+k \lceil \log_2 p\rceil}$\\
\hline
\hline
Axpy    & $0$  & $k$ \\
Div     & $2k-1$ & $0$\\
Table   & $0$  & $3$\\
Red     & $\geq 5k$  & $1$\\
\hline
  \end{tabular}
\caption{Complexity of the back and forth conversion between
  extension field and floating point numbers}\label{tab:comp}
\end{center}
\end{table}

\begin{figure}[ht]\center
\includegraphics[width=\textwidth*3/4]{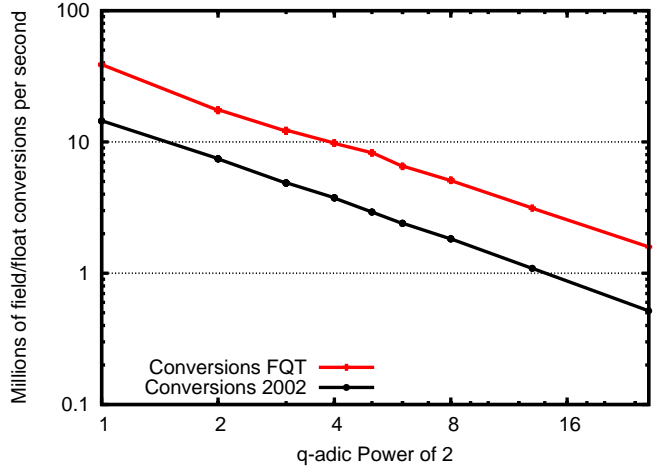}
\caption{Small extension field conversion speed on a Xeon 3.6GHz}\label{fig:conv}
\end{figure}
Figure~\ref{fig:conv} shows the speed of the conversion
after the floating point operations. The log scales prove
that for $q$ ranging from $2^1$ to
$2^{26}$ our new implementation is two to three
times as fast as the previous one\footnote{On a 32 bit Xeon.}.

\begin{figure*}[hbt]\center
\includegraphics[width=\textwidth*10/11]{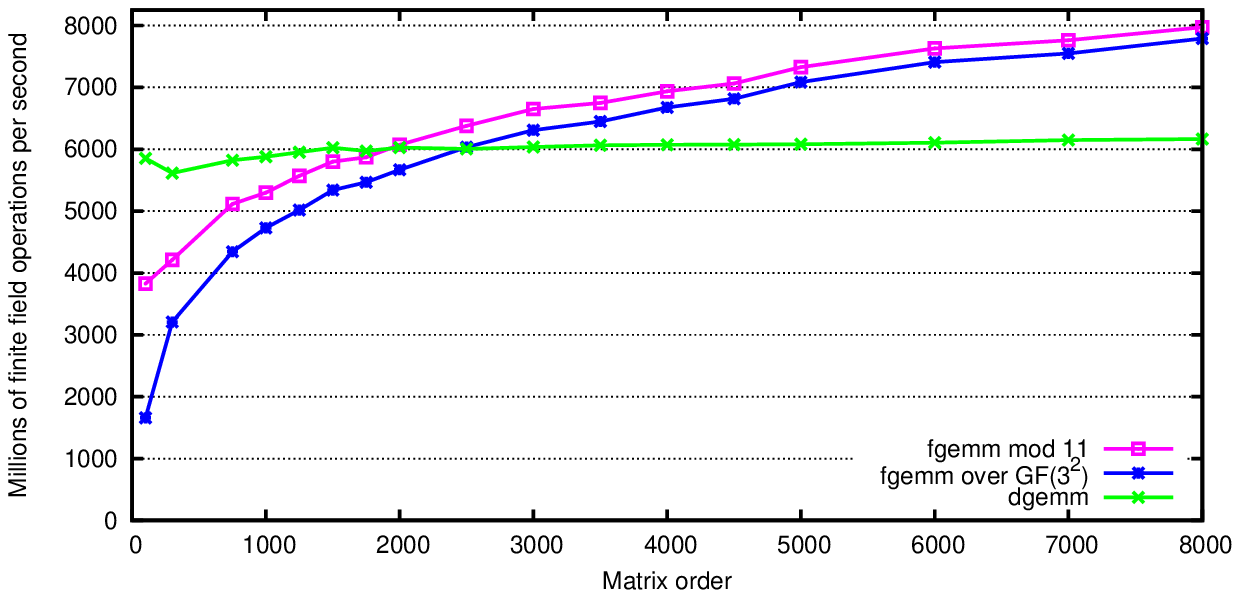}
\caption{Speed of finite field Winograd matrix multiplication on a XEON, 3.6 GHz}\label{fig:gfqgemm}
\end{figure*}
Furthermore, these improvements
allow the extension field routines to reach the
speed of 7800 millions of $\GF{9}$ operations per second\footnote{
On a XEON,                                                         
3.6 GHz, using Goto BLAS-1.09 {\tt dgemm} as the numerical routine
\citep{2002:gotoblas} and FFLAS {\tt fgemm} for the
fast prime field matrix multiplication \citep{jgd:2009:toms}.}
 as shown
on Figure~\ref{fig:gfqgemm}\footnote{The FFLAS routines are
available within the LinBox 1.1.4 library \citep{LinBox:2007} and the
\fqt is in implemented in the {\tt givgfqext.h} file of the
Givaro 3.2.9 library \citep{Givaro:2007}.}.
The speed-up obtained with these new implementations in also shown on this Figure.
It represents a reduction from the 15\% overhead of the
previous implementation to less than 4\% now, when compared over
$\GF{11}$.


\section{Application 3: Compressed Modular Matrix Multiplication}\label{sec:cmm}
We now extend the results of \cite{jgd:2008:mica} with the REDQ
algorithm. The idea is to use Kronecker substitution to pack several
matrix entries into a single machine word. We explore the
possibilities of packing on the left or on the right only, together
with packing on both matrices of a matrix multiplication.

\subsection{Middle Product Algorithm}
In this section, we show how
a dot product of vectors of size $d+1$ can be recovered from a
polynomial multiplication performed by a single
machine word multiplication.
This extends to matrix multiplication by compressing both matrices
first.
We first illustrate the idea for $2\times2$ matrices and~$d=1$.
The product
\[\begin{bmatrix}
  a & b \\
  c & d\\
\end{bmatrix}
\times
\begin{bmatrix}
  e & f\\
  g & h\\
\end{bmatrix}
=
\begin{bmatrix}
 ae+bg  &  af+bh \\
 ce+dg  &  cf+dh \\
\end{bmatrix}\]
is recovered from
\[
\begin{bmatrix}
  Qa+b \\
  Qc+d\\
\end{bmatrix}
\times
\begin{bmatrix}
  e+Qg & f+Qh\\
\end{bmatrix}
=
\begin{bmatrix}
 * + (ae+bg)Q + * \,Q^2  &  * + (af+bh)Q + * \,Q^2 \\
 * + (ce+dg)Q + * \,Q^2  &  * + (cf+dh)Q + * \,Q^2 \\
\end{bmatrix},
\]
where the character $*$ denotes other coefficients.

\begin{figure}[htbp]
\centerline{\includegraphics[width=\textwidth*4/11]{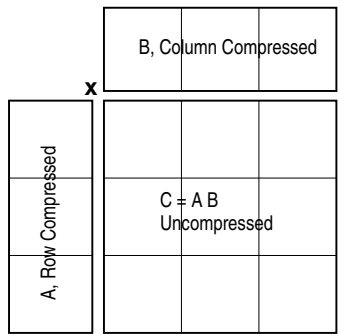}}
\caption{Compressed Matrix Multiplication (CMM)\label{fig:algo1}}
\end{figure}
In general, $A$ is an $m\times k$ matrix to be multiplied by a $k\times n$ matrix $B$,
the matrix $A$ is first compressed into a $m\times \left\lceil {k}/({d+1}) \right\rceil$
\verb!CompressedRowMatrix!, $CA$, and $B$ is transformed into a
$\left\lceil {k}({d+1}) \right\rceil \times n$
\verb!CompressedColumnMatrix!, $CB$. The compressed matrices are then multiplied and the result can be extracted from there. This is depicted on Fig.~\ref{fig:algo1}

In terms of number of arithmetic operations,
the matrix multiplication $CA \times CB$ can save \emph{a factor of
$d+1$} over the multiplication of $A\times B$ as shown on the $2\times 2$ case above.

The computation has three stages: compression, multiplication and extraction of the result. The compression and extraction are less demanding in terms of asymptotic complexity, but can still be noticeable for moderate sizes. For this reason, compressed matrices are often reused and it is more informative to distinguish the three phases in an analysis. This is done in Section~\ref{sec:comparison} (Table~\ref{tab:gains}), where
the actual matrix multiplication algorithm is also taken into account.


\paragraph*{Partial compression.}
Note that the last column of $CA$ and the last row of $B$ might not have $d+1$
elements if $d+1$ does not divide $k$. Thus one has to artificially
append some zeroes to the converted values. On $CB$ this
means just do nothing. On~$CA$ whose compression is reversed,
this means multiplying by $Q$ several
times.


\subsection{Available Mantissa and Upper Bound on $Q$ for the Middle Product}
If the product $CA\times CB$ is performed with floating point
arithmetic we just need that the coefficient of degree $d$ fits
in the $\beta$ bits of the mantissa. Writing $CA\times CB = c_H Q^d + c_L$, we see that this
implies that $c_H$, \emph{and only $c_H$}, must have entries that remain smaller than~$2^\beta$. It can then be recovered exactly by multiplication of $CA\times CB$ with
the correctly precomputed and rounded inverse of $Q^d$ 
as shown e.g., in \citep[Lemma~2]{jgd:2008:issac}.

With delayed reduction this means that
\[\sum_{i=0}^d \frac{k}{d+1} (i+1)(p-1)^2Q^{d-i} < 2^\beta.\]
On the other hand, delay reduction requires  (cf.~Eq~\eqref{eq:bounds})
\begin{equation}\label{eq:lower}k(p-1)^2<Q. \end{equation}
Thus the recovery is possible if
\begin{equation}\label{eq:upper} Q^{d+1}<2^\beta.\end{equation}
and a single reduction has to be made at the end of the dot product
as follows:
\begin{verbatim}
Element& init( Element& rem, const double dp) const {
        double r = dp;
        // Multiply by the inverse of Q^d with correct rounding
        r *= _inverseQto_d;
        // Now we just need the part less than Q=2^t
        unsigned long rl( static_cast<unsigned long>(r) );
        rl &= _QMINUSONE;
        // And we finally perform a single modular reduction
        rl %= _modulus;
        return rem = static_cast<Element>(rl);
}
\end{verbatim}

Note that one can avoid the
multiplication by the inverse of $Q$ when $Q$ is a power of~2, say~$2^t$:
by adding $Q^{2d+1}$ to the
final result one is guaranteed that the $t(d+1)$ high bits
represent exactly the $d+1$ high coefficients. On the one hand,
the floating point multiplication is then replaced by an addition.
On the other hand, this doubles the size of the dot product and thus
reduces by a factor of $\sqrt[d+1]{2}$ the largest possible dot
product size $k$.

\subsection{Middle Product Performance}

\begin{figure}[htb]
\centerline{\includegraphics[width=\textwidth]{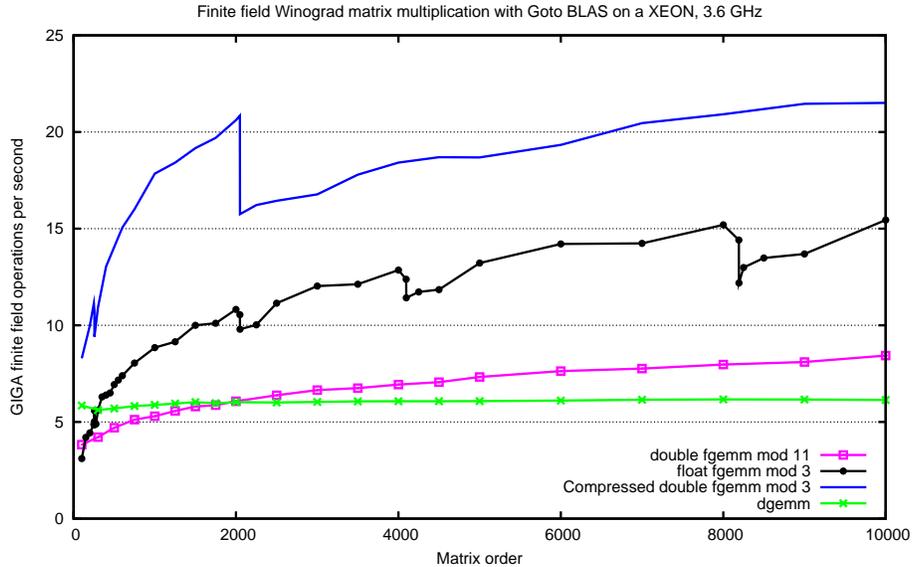}}
\caption{Compressed matrix multiplication 
compared with
\texttt{dgemm} (the floating point double precision
matrix multiplication of GotoBlas) and \texttt{fgemm} (the exact
routine of FFLAS) with double or single precision.
}\label{fig:compC}
\end{figure}
On Figure~\ref{fig:compC} we compare our compression algorithm to the
numerical double floating point
matrix multiplication \texttt{dgemm} of GotoBlas by \cite{2002:gotoblas}
and to the
\texttt{fgemm} modular matrix multiplication of the FFLAS-LinBox library by \cite{jgd:2002:fflas}. For the latter we show timings using \texttt{dgemm} and also \texttt{sgemm} over single floating points.
This figure shows that the compression ($d+1$) is very
effective for small primes: the gain over the double floating point
routine is quite close to $d$.

Observe that the curve of \texttt{fgemm} with underlying arithmetic on single
floats oscillates and drops sometimes. Indeed,
the matrix begins to be
too large and modular reductions are now required between the
recursive matrix multiplication steps.
Then the floating point
BLAS
routines
are used only when the sub-matrices are small
enough. One can see the subsequent increase in the number of classical
arithmetic steps on the drops around 2048, 4096 and 8192.

\begin{table}[htb]\center
\begin{tabular}{|c||c|c|c|c|c|c|c|c|c|}
\hline
Compression & 2 & 3..4 & 5..8 & 8 & 7 & 6 & 5 & 4 & 3\\
\hline
Degree d    & 1 & 5 &  9 & 7 & 6 & 5 & 4 & 3 & 2\\
\hline
Q-adic      & $2^3$ & $2^4$ & $2^5$ & $2^6$ & $2^7$ & $2^8$ & $2^{10}$ & $2^{13}$ & $2^{17}$\\
\hline
Dimensions  & $2$ & $\leq4 $& $\leq8$ & $\leq16$ & $\leq32$ & $\leq64$ & $\leq256$ & $\leq2048$ & $\leq32768$\\
\hline
\end{tabular}
\caption{Compression factors for different common matrix dimensions
  modulo~3, with $53$ bits of mantissa and $Q$ a power of
  $2$.}\label{tab:comprange}
\end{table}

On Table~\ref{tab:comprange}, we show the compression factors
modulo~3, with~$Q$ a power of~2 to speed up conversions. For a
dimension $n\le 256$ the compression is at a factor of five and the
time to perform a matrix multiplication is slightly more than a millisecond.
Then from dimensions from 257 to 2048 one has a factor of 4 and the
times are roughly 16 times the time of the four times smaller
matrix. The next stage, from 2048 to~32768 is the one that shows on
Figure~\ref{fig:algo1}.

Figure~\ref{fig:compC} shows the dramatic impact of the compression
dropping from $4$ to $3$ between $n=2048$ and $n=2049$.
It would be interesting to compare the multiplication of
$3$-compressed matrices of size $2049$ with a decomposition of the same
matrix into matrices of sizes $1024$ and $1025$, thus enabling
$4$-compression also for matrices larger than $2048$,
but with more modular reductions.

\subsection{Right or Left Compressed Matrix Multiplication}\label{sec:mixed} 

Another way of performing compressed matrix multiplication is to
multiply an uncompressed $m \times k$ matrix  to the right by a row-compressed
$k \times {n}/({d+1})$ matrix.
We illustrate the idea on~$2\times2$ matrices:
\[
\begin{bmatrix}
  a & b \\
  c & d\\
\end{bmatrix}
\times
\begin{bmatrix}
  e+Qf\\
  g+Qh\\
\end{bmatrix}
=
\begin{bmatrix}
(ae+bg)+Q(af+bh) \\
(ce+dg)+Q(cf+dh) \\
\end{bmatrix}
\]
\begin{figure}[htb]
\centerline{\includegraphics[width=\textwidth*4/11]{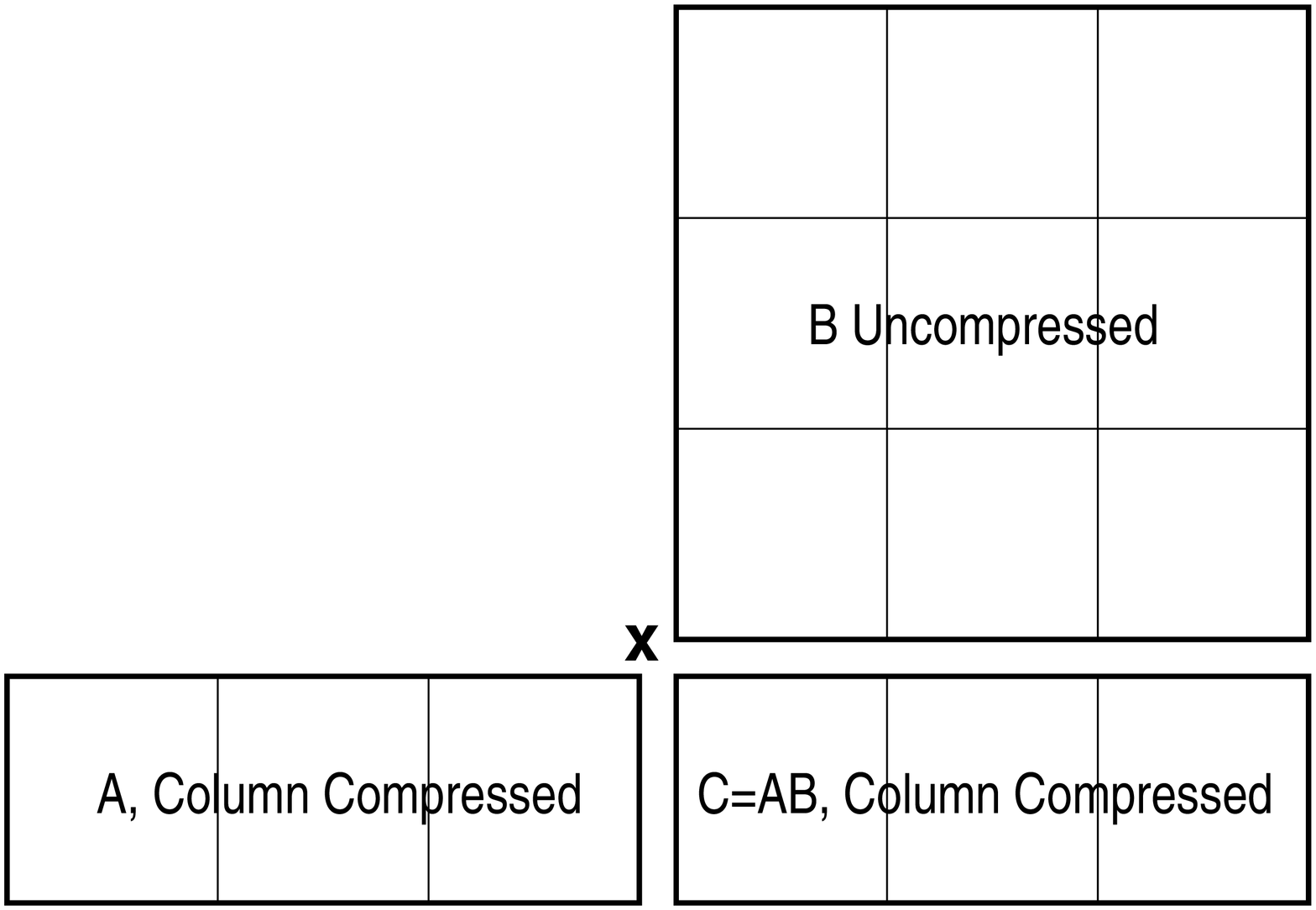}\hfill
\includegraphics[width=\textwidth*4/15]{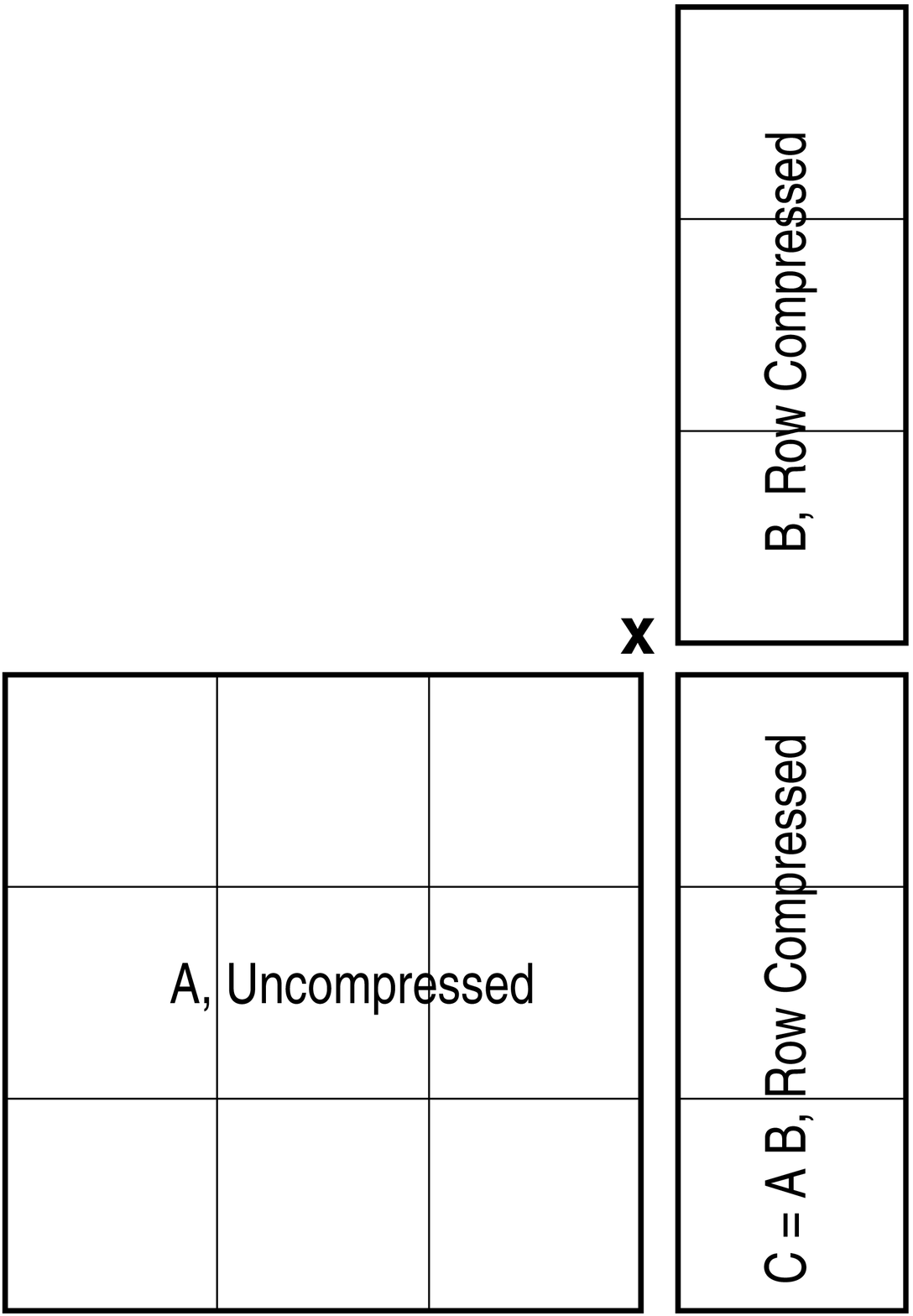}\hfill
\includegraphics[width=\textwidth*4/15]{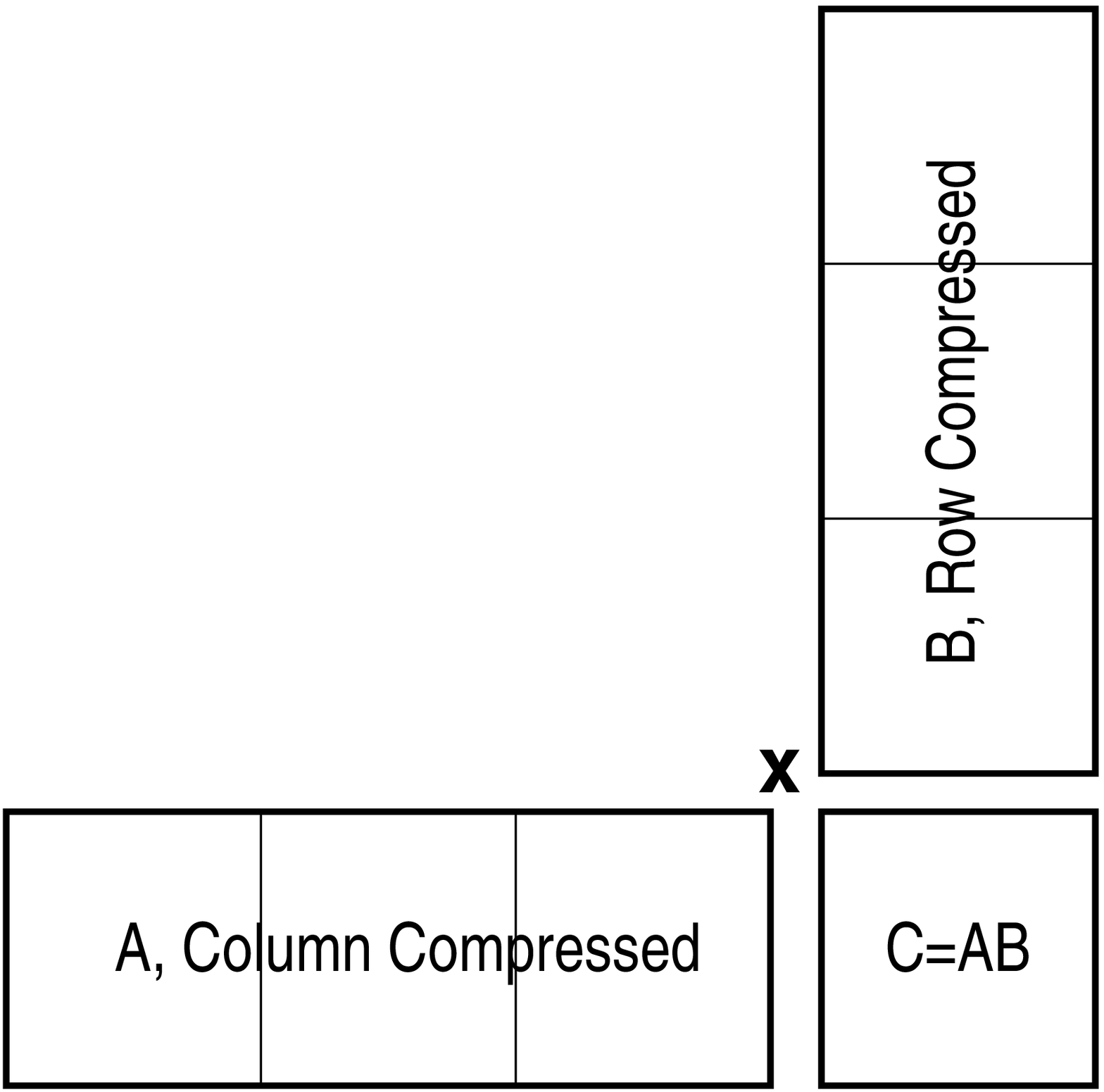}}
\caption{Left, Right and Full Compressions}\label{fig:leftrightfull}
\end{figure}
The general case is depicted on Fig.~\ref{fig:leftrightfull}, center. This is called Right Compressed Matrix Multiplication. Left Compressed Matrix Multiplication is obtained by transposition.
%
Here also $Q$ and $d$ must satisfy Eqs.~(\ref{eq:lower}) and
(\ref{eq:upper}).

The major difference with the Compressed Matrix Multiplication lies in the reductions.
Indeed, now one needs to
reduce simultaneously the $d+1$ coefficients of the polynomial in $Q$
in order to get the results. This simultaneous reduction can be made
by the REDQ algorithm. 

When working over compressed matrices~$CA$ and~$CB$, a first step is to uncompress~$CA$, which has to be taken into account when comparing methods.
Thus the whole right compressed matrix multiplication  is the following
algorithm 
\begin{equation}\label{eq:LeftComp} A=\operatorname{Uncompress}(CA); CC=A\times CB;
  \operatorname{REDQ}(CC)\end{equation}
\begin{figure}[htb]
\centerline{\includegraphics[width=\textwidth]{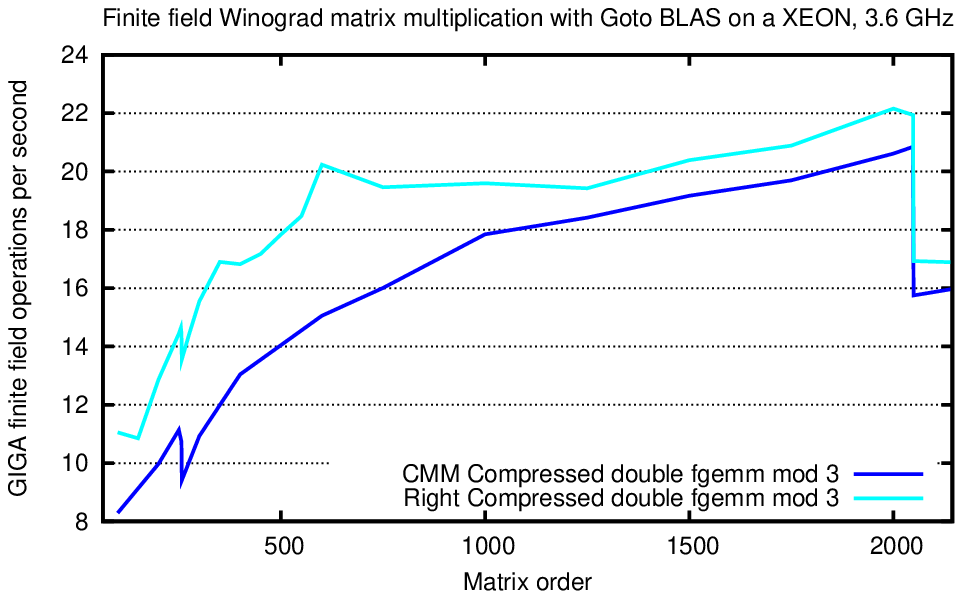}}
\centerline{\includegraphics[width=\textwidth]{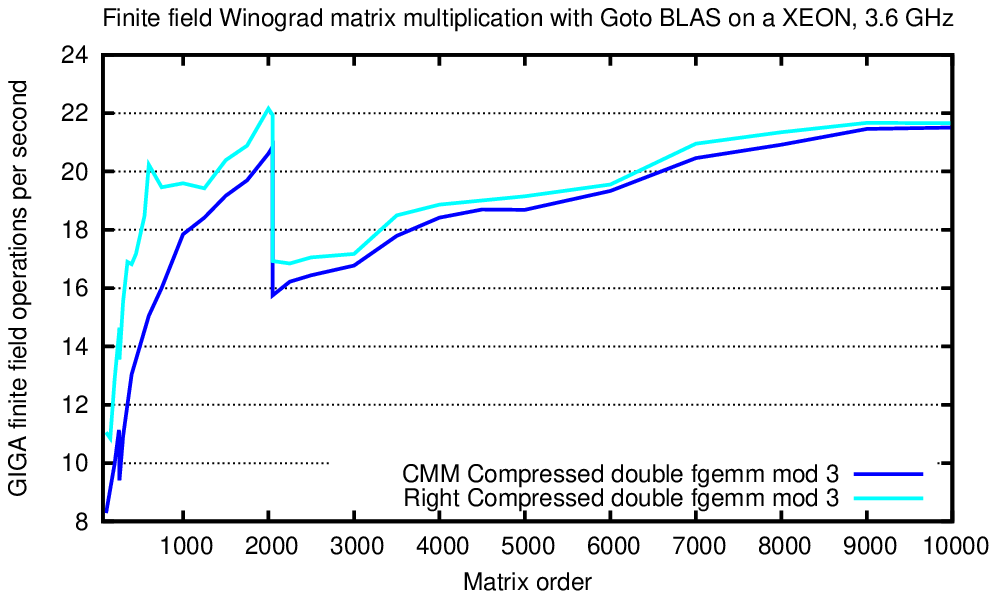}}
\caption{Right Compression and CMM.}\label{fig:right}
\end{figure}
We see on Figure~\ref{fig:right} that the used of REDQ instead of the
middle product algorithm has a high benefit. Indeed for small
matrices, the conversion can represent 30\% of the time and any
improvement there has a high impact.

\subsection{Full Compression} 
It is also possible to compress simultaneously both dimensions of
the matrix product (see Fig.~\ref{fig:leftrightfull}, right).
This is achieved by using 
polynomial multiplication
with two variables $Q$ and $\Theta$.
Again, we start by an example in dimension~2:
\[
\begin{bmatrix}
  a+Qc & b+Qd \\
\end{bmatrix}
\times
\begin{bmatrix}
  e+\Theta f\\
  g+\Theta h\\
\end{bmatrix}
=
\begin{bmatrix}
(ae+bg)+Q(ce+dg)+\Theta(af+bh)+Q\Theta(cf+dh) \\
\end{bmatrix}.
\]
More generally,
let $d_q$ be the degree in $Q$  and
$d_\theta$ be the degree in $\Theta$.
Then, the dot product is:
\begin{align*}
a\cdot b &=[\sum_{i=0}^{d_q} a_{i0}Q^i, \dots, \sum_{i=0}^{d_q} a_{in}Q^i] \times [\sum_{j=0}^{d_\theta} b_{0j}\Theta^j, \dots,
\sum_{j=0}^{d_\theta} b_{nj}\Theta^j],\\
&=\sum_{l=0}^k   (\sum_{i=0}^{d_q} a_{il})(\sum_{j=0}^{d_\theta}
b_{lj})Q^i\Theta^j = \sum_{i=0}^{d_q}\sum_{j=0}^{d_\theta}( \sum_{l=0}^k  a_{il}b_{lj})Q^i\Theta^j.
\end{align*}

In order to guarantee that all the coefficients can be recovered
independently, $Q$ must still satisfy Eq.~(\ref{eq:lower}) but then
$\Theta$ must satisfy an additional constraint:
\begin{equation}\label{eq:qtheta}
Q^{d_q+1} \leq \Theta.
\end{equation}
This imposes restrictions on $d_q$ and $d_\theta$:
\begin{equation}\label{eq:dq}
Q^{(d_q+1)(d_\theta+1)} < 2^\beta.
\end{equation}

\subsection{CMM Comparisons}\label{sec:comparison}

In Table~\ref{tab:gains}, we summarize the differences between the algorithms presented on Figures~\ref{fig:algo1} and~\ref{fig:leftrightfull}. As usual, the exponent~$\omega$ denotes the exponent in the complexity of matrix multiplication. Thus, $\omega=3$ for the classical matrix multiplication, while $\omega<3$ for faster matrix multiplications, as used in \citep[\S
3.2]{jgd:2009:toms}. For products of rectangular matrices, we use the classical technique of first decomposing the matrices into square blocks and then using fast matrix multiplication on those blocks.

\subsubsection{Compression Factor}
The costs in Table~\ref{tab:gains} are expressed in terms of a \emph{compression factor}~$e$, that we define as
\[ e:= \left\lfloor \frac{\beta}{\log_2(Q)} \right\rfloor,\]
where, as above, $\beta$ is the size of the mantissa and~$Q$ is the integer chosen according to Eqs.~\eqref{eq:lower} and~\eqref{eq:upper}, except for Full Compression where the more constrained Eq.~\eqref{eq:dq} is used.

Thus the degree of compression for the first three algorithms is just
$d=e-1$, while it becomes only $d=\sqrt{e}-1$ for the full compression algorithm
(with equal degrees $d_q=d_\theta=d$ for both variables $Q$ and
$\Theta$).

\begin{table}[ht]\center
\begin{tabular}{|c||c|c|c|c|}
\hline
Algorithm               & Operations    & Reductions & Conversions\\
\hline
CMM     & $\GO{ m n
  \left(\frac{k}{e}\right)^{\omega-2}}$ & $m \times n$ $\operatorname{REDC}$ &
$\frac{1}{e} mn$ $\operatorname{INIT}_e$\\
\hline
Right Comp.     & $\GO{  m k \left(\frac{n}{e}\right)^{\omega-2}}$ & $m
\times \frac{n}{e}$ $\operatorname{REDQ}_e$ & $\frac{1}{e}mn$ $\operatorname{EXTRACT}_e$\\
\hline
Left Comp.      & $\GO{  n k \left(\frac{m}{e}\right)^{\omega-2}}$ & $\frac{m}{e}
\times n$ $\operatorname{REDQ}_e$ & $\frac{1}{e}mn$ $\operatorname{EXTRACT}_e$\\
\hline
Full Comp. & $\GO{  k \left(\frac{mn}{e}\right)^{\frac{\omega-1}{2}} }$
& $\frac{m}{\sqrt{e}}
\times \frac{n}{\sqrt{e}}$ $\operatorname{REDQ}_e$ & $\frac{1}{e}mn$ $\operatorname{INIT}_e$ \\
\hline
\end{tabular}
\caption{Number of arithmetic operations for the different
  algorithms}\label{tab:gains}
\end{table}

\subsubsection{Analysis}
In terms of asymptotic complexity, the cost in number of arithmetic operations is dominated by that of the product (column Operations in the table), while reductions and conversions are linear in the dimensions.
This is well reflected in practice.
For example, with algorithm Right Compression on matrices of sizes
$10,000\times 10,000$ it took $90.73$ seconds to perform the matrix
multiplication modulo $3$ and $1.63$ seconds to convert the resulting
matrix. This is less than $2$\%.
For $250\times 250$ matrices it takes less than $0.00216$ seconds to
perform the multiplication and roughly $0.0016$ seconds for the
conversions. There, the conversions account for $43\%$ of the time and
it therefore of extremely high importance to optimize the conversions.

In the case of rectangular matrices, the second column of Table~\ref{tab:gains} shows that one should choose the algorithm depending on the largest dimension: CMM if the common dimension~$k$ is the largest, Right Compression if~$n$ if the largest and Left Compression if~$m$ dominates. The gain in terms of arithmetic operations is
$e^{\omega-2}$ for the first three variants and
$e^{\frac{\omega-1}{2}}$ for full compression. This is not only of
theoretical interest but also of practical value, since the compressed matrices are then less rectangular. This enables more locality for the matrix computations and usually results in better performance.
Thus, even if~$\omega=3$, i.e., classical multiplication is used, these considerations point to a source of speed improvement.

The full compression algorithm seems to be the best candidate for
locality and use of fast matrix multiplication; however the
compression factor is an integer, depending on the flooring of either
$\frac{\beta}{\log_2(Q)}$ or $\sqrt{\frac{\beta}{\log_2(Q)}}$. Thus
there are matrix dimensions for which the compression factor of
e.g., the right compression will be larger than the square of the
compression factor of the full compression. There the right
compression will have some advantage over the full compression.

If the matrices are square ($m=n=k)$ or if $\omega=3$, the products all become the same, with similar constants implied in the~$O()$, so that apart from locality considerations, the difference between them lies in the time spent in reductions and conversions. Since the $\operatorname{REDQ}_e$ reduction is faster than $e$ classical
reductions~\citep{jgd:2008:issac}, and since $\operatorname{INIT}_e$ and
$\operatorname{EXTRACT}_e$ are roughly the same operations, the best algorithm would
then be one of the Left, Right or Full compression.
Further work would include implementing the Full
compression and comparing the actual timings of conversion overhead
with that of the Right algorithm and that of CMM.



\section{Conclusion}
We have proposed a new algorithm for simultaneous reduction of several
residues stored in a single machine word.
For this algorithm we also give a time-memory trade-off implementation
enabling very fast running time if enough memory is available.

We have shown very effective applications of this trick for packing
residues in large applications. This proves efficient for
modular polynomial multiplication, extension fields conversion to
floating point and linear algebra routines 
over small prime fields.

Further work is needed to compare of running times between
different choices for $q$. Indeed our experiments were made with $q$ a
power of two and large table look-up. With $q$ a multiple of $p$ the
table look-up is not needed but divisions by $q^i$ will be more
expensive. A possibility would be taking $q$ in the form $q=p2^t$, 
then only divisions  by $p$ or $p^i$ would be made.

It would also be interesting to see in practice how this trick extends 
to larger precision implementations: on the one hand the basic
arithmetic slows down, but on the other hand the trick enables a more
compact packing of elements (e.g., if an odd number of field elements can
be stored inside two machine words, etc.).



\nocite{Harvey:2007:kronecker,Boldo:2008:delayed}
\bibliographystyle{elsart-harv}
\bibliography{cmm}

\end{document}